\documentclass[aps,prl,twocolumn,superscriptaddress,]{revtex4-2}
\usepackage{graphics,graphicx,array,afterpage}
\usepackage{grffile}
\usepackage[export]{adjustbox}
\usepackage{lineno}
\usepackage{xcolor}
\usepackage{longtable}
\usepackage{svg}
\usepackage{braket}
\usepackage{multirow}
\usepackage{enumerate}
\usepackage[colorlinks, citecolor=blue, urlcolor=blue]{hyperref}
\usepackage{amsmath}
\usepackage{amssymb}
\usepackage{amsthm}
\usepackage{pifont}

\newtheorem{theorem}{Theorem}
\newtheorem{corollary}{Corollary}
\newtheorem{lemma}{Lemma}
\newtheorem{definition}{Definition}

\newcommand{\tr}{\mbox{tr}}

\newcommand{\SWAP}{{\rm SWAP}}
\newcommand{\ipic}[3][-0.45]{\raisebox{#1\height}{\scalebox{#3}{\includegraphics{#2}}}}

\begin{document}

\title{No-Free-Lunch Theories for Tensor-Network Machine Learning Models}
\author{Jing-Chuan Wu}
\thanks{These authors contributed equally to this work.}
\affiliation{Theoretical Physics Division, Chern Institute of Mathematics and LPMC, Nankai University, Tianjin 300071, China}
\author{Qi Ye}
\thanks{These authors contributed equally to this work.}
\affiliation{Center for Quantum Information, IIIS, Tsinghua University, Beijing 100084, China}
  \affiliation{Shanghai Qi Zhi Institute, Shanghai 200232, China}
\affiliation{School of Engineering and Applied Sciences, Harvard University, Cambridge, Massachusetts 02138, USA}
\author{Dong-Ling Deng}
\affiliation{Center for Quantum Information, IIIS, Tsinghua University, Beijing 100084, China}
  \affiliation{Shanghai Qi Zhi Institute, Shanghai 200232, China}
 \affiliation{Hefei National Laboratory, Hefei 230088, China}
\author{Li-Wei Yu}
\email{yulw@nankai.edu.cn}
\affiliation{Theoretical Physics Division, Chern Institute of Mathematics and LPMC, Nankai University, Tianjin 300071, China}
\begin{abstract}
Tensor network machine learning models have shown remarkable versatility in tackling complex data-driven tasks, ranging from quantum many-body problems to classical pattern recognitions. Despite their promising performance, a comprehensive understanding of the underlying assumptions and limitations of these models is still lacking. In this work, we focus on the rigorous formulation of their no-free-lunch theorem---essential yet notoriously challenging to formalize for specific tensor network machine learning models. In particular, we rigorously analyze the generalization risks of learning target output functions from input data encoded in tensor network states. We first prove a no-free-lunch theorem for machine learning models based on matrix product states, i.e., the one-dimensional tensor network states. Furthermore, we circumvent the challenging issue of calculating the partition function for two-dimensional Ising model, and prove the no-free-lunch theorem for the case of two-dimensional projected entangled-pair state, by introducing the combinatorial method associated to the ``puzzle of polyominoes''.  Our findings reveal the intrinsic limitations of tensor network-based learning models in a rigorous fashion, and open up an avenue for future analytical exploration of both the strengths and limitations of quantum-inspired machine learning frameworks.

\end{abstract}


\maketitle

Tensor networks (TNs) have emerged as a powerful tool for studying quantum many-body systems, demonstrating remarkable versatility across various domains of quantum physics \cite{Orus2019Tensor, Biamonte2019Lectures, Cirac2020Matrix, Banuls2023Tensora,Rieser2023Tensor, Wang2023Tensor}. This success has catalyzed a growing interest in harnessing TNs for machine learning applications \cite{Cichocki2014Tensor, Novikov2015Tensorizing,Cichocki2016Tensor, Cichocki2016Tensor2, Stoudenmire2016Supervised, Novikov2016Exponential, Liu2018Entanglement,Chen2018Parallelized,   Levine2018Deep, Stoudenmire2018Learning, Han2018Unsupervised,Liu2019Machine,Hayashi2019Exploring, Liu2021Tensor, Chen2018Equivalence, Levine2019Quantum, Bhatia2019Matrix, Efthymiou2019Tensornetwork, Huggins2019Towards, Glasser2020From,Su2020Convolutional,Sun2020Tangent, Sun2020Model, Chen2020Hybrid, Sun2020Generative,Wang2020Anomaly,Gao2020Compressing, Wall2021Generative, Cheng2021Supervised, Costa2021Tensortrain, Kardashin2021Quantum,Felser2021Quantuminspired,Wang2021Experimental,Hawkins2021Bayesian, Chen2022Kernelized,Vieijra2022Generative,Shi2022Clustering,Convy2022Mutual, Metz2023Selfcorrecting,Liao2023Decohering,Ran2023Tensor,Wu2023Tensornetwork,Lopez-Piqueres2023Symmetric,Meng2023Residual,Shin2024Dequantizing, Wesel2024Tensor, Bhatia2024Federated, Tomut2024CompactifAI, Teng2024Learning, Bermejo2024Quantum,Casagrande2024TensorNetworksbased,Chen2024Machine,Su2024Language}, where they have shown promise in diverse areas such as dimensionality reduction \cite{Cichocki2016Tensor,Cichocki2016Tensor2}, model compression \cite{Sun2020Model,Tomut2024CompactifAI}, natural language processing \cite{Guo2018Matrix,Meichanetzidis2020Quantum,Su2024Language}, generative models \cite{Han2018Unsupervised,Lopez-Piqueres2023Symmetric}. In particular, TNs have demonstrated distinct advantages in the development of efficient, interpretable machine learning models \cite{Ran2023Tensor}.  Meanwhile, TN-based machine learning models offer intriguing theoretical and practical advantages, including insights into learning theory and the development of quantum machine learning frameworks. In addition, TNs provide a powerful framework for developing quantum machine learning theories, showing potential exponential advantages over those classical models \cite{Gao2018Quantum, Levine2019Quantum, Gao2021Enhancing}. Recent advancements also highlight the use of TNs to mitigate barren plateaus in quantum learning models \cite{Rudolph2023Synergistica}. On the practical side, TN techniques, including canonical form and renormalization \cite{Stoudenmire2016Supervised, Stoudenmire2018Learning}, are instrumental in optimizing and training machine learning models, further supported by the development of open-source libraries \cite{Abadi2016TensorFlow, 
Roberts2019tensornetwork}. In the vein of TN-based machine learning theory, a number of pioneering works have been conducted, including these on the barren plateau problem \cite{Liu2022Presence,Liu2023Theory,Garcia2023Barren, Miao2024Isometric, Barthel2024Absence}, model generalization \cite{Strashko2022Generalization}, and mutual information scaling \cite{Convy2022Mutual}. Despite this rapid progress, the field of TN-based machine learning is still in its infancy, with many fundamental aspects remaining largely unexplored. Here, we study the generalization ability of TN-based machine learning models, with a focus on establishing their no-free-lunch (NFL) theorem.

\begin{figure}
\centering
\includegraphics[width=\linewidth]{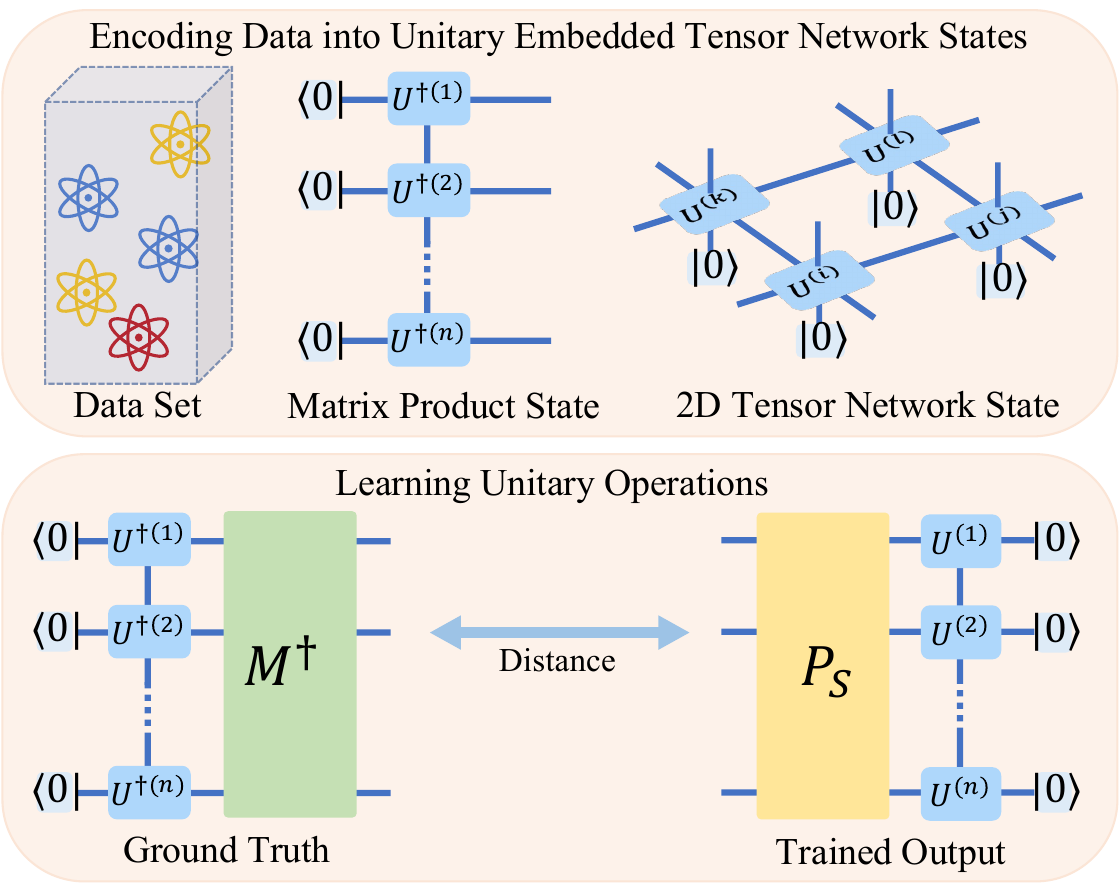} 
\caption{Schematic illustration of supervised learning  unitaries based on tensor network states. Upper panel: The encoding strategy.  Classical (Quantum) data samples with labels are encoded into the local tensors $U^{(i)}$ of the unitary embedded tensor network states. Lower panel: The learning strategy. Given the training set $\mathcal{S}$ of samples with the labeled outputs (left), the  goal is to minimize the average distance between the learned output states and the ground truth states (acting the unitary $M^\dag$ on encoding states) over all training samples. Unitary circuit $P_{\mathcal{S}}$ stores the variational parameters.}
\label{Schematic_Figure}
\end{figure}

The NFL theorem is one of the most fundamental theorems in the classical machine learning theory \cite{Schaffer1994Conservation,Wolpert1996Lack,Wolpert1997No}. It states that, averaged over all possible problems, every algorithm performs equally well when applied to problems they were not specifically designed for. Consequently, the average performance of a machine learning model across all target functions is highly dependent on the size of the input set, highlighting the importance of large datasets in building robust models, such as large language models. Inspired by the critical role of NFL theorem in classical machine learning, significant progress has been made in developing NFL theorem for quantum learning models \cite{Poland2020No,Volkoff2021Universal,Sharma2022Reformulation,Brierley2022No,Zhao2023Learning,Wang2024Transition,Wang2024Separable}. The quantum NFL theorem establishes straightforward connections between quantum features and the capabilities of quantum learning models.
For instance, in practical quantum learning setups with a finite number of measurements, entangled data exhibits a dual effect on prediction errors. With sufficient measurements, highly entangled data can reduce prediction errors. This is consistent with the ideal case of infinite measurements \cite{Sharma2022Reformulation}. Conversely, with few measurements, highly entangled data can amplify predicting errors \cite{Wang2024Transition}. These results highlight how quantum features contribute to the advantages in quantum machine learning models.

For TN-based machine learning models,  the classical data is encoded in TN states with specific structures, and the corresponding data distribution can be quite different from those in other machine learning models. It is hence highly desirable to investigate how to rigorously formulate the NFL theorem for the specific TN-based models, as well as to study how the bond and physical dimensions of TNs would influence the lower bounds of the average risks. We face two major challenges in establishing the NFL theorem for higher-dimensional TN models: The first one is to calculate the variance of random TN states.  This can be mapped to calculating the partition function of a high-dimensional Ising model, which is notoriously a challenging task.   The second one is to embed the information learned from the training set (corresponding to the diagonalized sectors of the matrix $M^\dag P_{\mathcal S}$ in Fig.~\ref{Schematic_Figure}) into the higher-order tensor properly, which notably differs from the case for quantum learning models \cite{Poland2020No,Sharma2022Reformulation} and requires judicious design. In this paper, we tackle these challenges and rigorously prove the NFL theorem for TN models. Our setup is to learn a target unitary operator through data samples encoded into TN states, see Fig.~\ref{Schematic_Figure} for illustration. We first prove the NFL theorem for the case of 1D matrix product state (MPS). We show that the lower bound of the average predicting risk  highly depends on the size of linearly independent training samples. We then focus on the 2D projected entangled-pair states (PEPS). We note that proving the NFL theorem for 2D PEPS poses significant challenges, as it equates to calculating the partition function of the 2D Ising model, which is a notoriously difficult problem.  Here, we circumvent this obstacle by introducing a combinatorial method associated to the ``puzzle of polyominoes'' \cite{Liu2023Theory}. In addition, we carry out numerical simulations to support our analytical results.

{\it Tensor network machine learning.}--- We consider a task of learning the unknown unitary operation $M$ based on the input of TN states.  Without loss of generality, we take the 1D MPS for demonstration. The MPS under periodic boundary condition has the form $|\psi\rangle = \sum_{i_1,\dots , i_n }\tr [A_{i_1}^{(1)}A_{i_2}^{(2)}\cdots A_{i_n}^{(n)}]|i_1, i_2, \dots ,i_n\rangle$, where  $A_{i_k}^{(k)}$ denotes the $D\times D$ tensor with $D$ representing the bond dimension,  $|i_k\rangle$ denotes the state of $k$-th physical site with physical dimension $d$. We define the unitary embedded MPS by converting each $D\times D\times d$ tensor $A^{(k)}$ to a $Dd\times Dd$ unitary $U^{(k)}$ \cite{Perez2007Matrix,Gross2010Quantum,Silvano2010Statistical, Silvano2010Typicality,collins2012matrix,Haferkamp2021Emergent},  as depicted in Fig.~\ref{Schematic_Figure}.  We define the labeled training set $\mathcal{S} = \{( |\psi_j\rangle, |\phi_j\rangle)|j=1,2,...t\}$, where the site size $|\mathcal{S}|=t$, the  MPS $|\psi_j\rangle$ belongs to the feature Hilbert space, and the state $|\phi_j\rangle =M |\psi_j\rangle$ belongs to the label Hilbert space. We learn the target unitary $M$  by minimizing the loss function $\mathcal{L}=\sum_{j=1}^t \left|\left(M |\psi_j\rangle-P_{\mathcal{S}}|\psi_j\rangle\right)/\langle\psi_j|\psi_j\rangle\right|^2$, where  $P_{\mathcal{S}}$ denotes the variational quantum circuit with sufficient expressivity.   If the model is properly trained, then one has $P_{\mathcal{S}}|\psi_j\rangle = M|\psi_j\rangle,\, \forall |\psi_j\rangle \in \mathcal{S}$ up to an overall phase. 
To quantify the predicting accuracy of our TN model on learning the target $M$, we define the predicting risk function by the following trace-norm formula 
\begin{equation}\label{Risk_error}
R_M(P_{\mathcal{S}}) = \int dx\left\| {\left( M|x\rangle\langle x|M^\dag-P_{\mathcal{S}}|x\rangle\langle x|P_{\mathcal{S}}^\dag\right)}/{\langle x| x\rangle}\right\|_1^2,
\end{equation}
where $\|A\|_1 = \frac{1}{2}\tr[\sqrt{A^\dag A}]$ denotes the trace norm of $A$, $|x\rangle$ represents the unitary embedded MPS, and the integral is over the Haar measure of all local unitary tensors $\{U^{(i)}|i=1,2,\cdots n\}$ in Fig.~\ref{Schematic_Figure}.  $R_M(P_{\mathcal{S}})$ represents the prediction error of the trained model. For proper learning without training errors, $R_M(P_{\mathcal{S}})$ is equivalent to the generalization error \cite{Caro2022Generalization}. 
We note that the norm $ \langle x|x\rangle$ is exponentially concentrated around one \cite{Haferkamp2021Emergent}. The  risk function describes the  probability that $P_S$ fails to reproduce the target unitary $M$ with the random MPS as the testing input \cite{Poland2020No}.  One can extend the above MPS-based learning model to higher dimensional tensor-network based models, by replacing $|\psi\rangle$ and $|x\rangle$ into the higher dimensional TN states.  With the above risk function, one can then study the  NFL theorem for TN-based models in learning a target unitary operation $M$.  

{\it The 1D case.}--- The MPSs have found broad applications in machine learning, particularly in the context of learning and representing complex patterns and data, including dimensionality reduction, generative models, and so on.  Now we consider the NFL theorem for the case of MPS. To showcase the NFL theorem, we adopt the risk function $R_M(P_{\mathcal{S}})$ and calculate the average risk over arbitrary target unitary $M$ and arbitrary training set $\mathcal{S}$. Here for simplicity, we suppose that the states in the training set are linearly independent \cite{Sharma2022Reformulation}.
\begin{theorem}\label{theorem:1D} 
Define the risk function $R_M(P_{\mathcal{S}})$ in Eq.~(\ref{Risk_error}) for learning a target $n$-qubit unitary $M$ based on the input of MPSs, where $P_{\mathcal{S}}$ represents the hypothesis unitary learned from the training set $\mathcal{S}$. Given a linear independent training set with size $t_k = d^n-d^{n-k}$, the integer $k\in [1,n-1]$, $d$ is the physical dimension of MPS, and $n$ denotes the qubit number of the system. The average risk is lower bounded by
\begin{equation}\label{NFL_Thm1}
\begin{aligned}
\mathbb{E}_{M,\mathcal{S}}& \left[R_M(P_{\mathcal{S}})\right]\geq  1-(1-\frac{2}{d^k})(1+(dAB)^n)\\
&-(\frac{1}{d^n}+\frac{1}{d^k})(A^k+B^k)(1+(dAB)^{n-k}),
\end{aligned}
\end{equation} 
where  $A=\frac{D+1}{Dd+1}$, $B=\frac{D-1}{Dd-1}$, and $D$ is the bond dimension of MPS. 
\end{theorem}
\begin{proof}
We sketch the main idea here and leave the detailed proof to the Supplemental Materials \cite{NFLTN_Supplementary}.  Since the norm $ \langle x|x\rangle$ in Eq.~(\ref{Risk_error}) is exponentially concentrated around one \cite{Haferkamp2021Emergent}, one simplifies the risk function $R_M(P_{\mathcal{S}})$ to be 
$R_M(P_{\mathcal{S}}) \rightarrow 1-\int dx \left|\langle x|M^\dag P_{\mathcal{S}}|x\rangle\right|^2$. For the convenience of analytical calculations,  we consider the  training set $\mathcal{S}$ with size  $|\mathcal{S}|=t_k=d^n-d^{n-k}$. For proper learning where $P_{\mathcal{S}}|\psi_j\rangle = M|\psi_j\rangle,\, \forall |\psi_j\rangle \in \mathcal{S}$ up to an overall phase, we consider the case that the learned unitary $P_{\mathcal{S}}$ obeys  $M^\dag P_{\mathcal{S}} =W=e^{i\theta I_{t_k}}\oplus Y$, where the global phase term $e^{i\theta I_{t_k}}$ locates in the space of the training set, $I_{t_k}$ denotes the $t_k\times t_k$ identity matrix, and the $d^{n-k}\times d^{n-k}$  unitary matrix $Y$ locates in the complementary subspace. For simplicity, we suppose that  the unitary $Y$ locates in an $(n-k)$-qudit subsystem. Utilizing the property that each local random unitary $U^{(i)}$ of the unitary embedded MPS $|x\rangle$ constitutes approximately unitary $2$-design \cite{Haferkamp2021Emergent,Liu2022Presence}, one can map the 2-moment integral of $|x\rangle$ into calculating the partition function of 1D classical Ising model. Averaging over all possible $M$ and all training sets $\mathcal{S}$,  one obtains the analytical lower bound of  $\mathbb{E}_{M,\mathcal{S}}[R_M(P_{\mathcal{S}})]$. This leads to Eq.~(\ref{NFL_Thm1}) and completes the proof.
\end{proof}

Theorem~\ref{theorem:1D} establishes an analytical lower bound  for the average risk of the MPS-based machine learning model, and thus quantifies the capability of the model in learning an arbitrary target unitary with an arbitrary training set.  Eq.~(\ref{NFL_Thm1}) indicates that  the average risk depends solely on the training set size $|\mathcal{S} |= t_k$. Through the numerical calculations,  we show that the average risk decreases monotonically with respect to the increasing  $|\mathcal{S} |$. Rigorous proofs show that the average risk is lower bounded by zero for the full training set, whereas the average risk is lower bounded by one for the empty training set. These results formalize the  NFL theorem for MPS-based machine learning models. Apart from the case of learning arbitrary target unitaries from the input of MPSs, we also consider the case of learning matrix product operators from quantum states, and analytically formulate the corresponding NFL theorem \cite{NFLTN_Supplementary}.

{\it The 2D case.}--- 2D TN states present notable advantages for studying quantum many-body systems and computational applications \cite{Verstraete2006Criticality,Jordan2008Classical,Gu2009Tensor,Buerschaper2009Explicit,Piroli2020Quantum,Haghshenas2022Variational,Pan2022Simulation,Pan2022Solving,Wang2024Tensor}. They enable efficient representation of large quantum systems with reduced computational costs, respecting the area law for entanglement entropy \cite{Eisert2010Colloquium}. 
However, the contraction of a general PEPS without any canonical form is a $\#$P-hard problem \cite{Schuch2007Computational}, where the complexity of computing different physical properties, e.g., the norm and expectation value, shall grow exponentially with respect to the system size. Here we formulate the NFL theorem for the PEPS-based machine learning models.

\begin{theorem}\label{theorem:2D} 
Define the risk function $R_M(P_{\mathcal{S}})$ in Eq.~(\ref{Risk_error}) for learning a target $L^2$-qubit unitary $M$ based on the input of PEPS, where $P_{\mathcal{S}}$ represents the hypothesis unitary learned from the training set $\mathcal{S}$. Given a linear independent training set with size $t_k = d^{L^2}-d^{L^2-k}$, the integer $k\in [1,L^2-1]$, $d$ represents the physical dimension and virtual dimension of PEPS, and $L^2$ denotes the qubit number of the system. The average risk is lower bounded by
\begin{equation}\label{NFL_Thm2}
\begin{aligned}
&\mathbb{E}_{M,\mathcal{S}} \left[R_M(P_{\mathcal{S}})\right] \geq 1-(1+c(0.7)^L)\left[1-\frac{2}{d^{k}}+(1+\frac{1}{d^{L^2-k}})\right.\\
&\left.\cdot\left(\frac{2D^4 d-2}{D^4 d^3-d}\right)^{k}\left(\frac{1+D}{2D}\right)^{2k}(1+G(1/d,1/D^2))^{2l}\right],
\end{aligned}
\end{equation} 
where  $l=\lceil \sqrt{k} \rceil$, $G(q,p)=\frac{p}{2}\left(\sqrt{\frac{(1+q)(1+q-qp)}{1-q(2+p)+q^2(1-p)}}-1\right)$, $D$ is the bond dimension of the PEPS and $c$ is a constant. 
\end{theorem}

\begin{proof}
We outline the main idea here and leave the details to the Supplemental Materials \cite{NFLTN_Supplementary}. For the 2D case, one critical step  is to map the calculation of the second moment for 2D uniatry embedded PEPS into calculating the partition functions of 2D classical Ising models, which is typically a challenging problem. To address this, we introduce the combinatorial method based on the ``puzzle of polyominoes'' \cite{Liu2023Theory}, and thus convert the problem of calculating partition functions into enumerating the directed polyominoes in a planar graph \cite{Bousquet-Melou1998New}. Then by utilizing the generating function $G(p,q)$ for enumerating the directed polyominoes, one can efficiently compute the value of the second moment integral of 2D unitary embedded PEPS. Consequently, one obtains the analytical lower bound of the average risk function $\mathbb{E}_{M,\mathcal{S}}[R_M(P_{\mathcal{S}})]$. This leads to Eq.~(\ref{NFL_Thm2}) and completes the proof. 
\end{proof}

Theorem~\ref{theorem:2D} establishes a lower bound for the average risk function of PEPS-based machine learning models. It indicates that, in the task of learning an arbitrary unitary with the input of PEPS, when averaged over arbitrary training datasets and learning models, the generalization risk of a PEPS-based learning model depends solely on the size $t_k$ of the training set. As the size $t_k$ increases, the average risk decreases, eventually reaching zero when the training set is complete. These results thus rigorously formalize the NFL theorem for the PEPS-based machine learning models.

Our NFL theorems for TN-based models are different from those for classical machine learning models. In classical supervised learning, to learn a target function $f$ based on the labeled training set $S =\{(x_i,y_i)| x_i\in \mathcal{X}, y_i =f(x_i)\in \mathcal{Y}\}_{i=1}^t$, one typical NFL formulation is $\mathbb{E}_{f,S}\left[R_f\left(h_S\right)\right] \geq\left(1-{1}/{|\mathcal{Y}|}\right)\left(1-{t}/{|\mathcal{X}|}\right)$, where $|\mathcal{X}|$ $(|\mathcal{Y}|)$ represents the dimension of $\mathcal{X}$ $(\mathcal{Y})$, and the risk function $R_f(h_S) = \textrm{Pr}[h_{S}(x)\neq f(x)]$ means the probability that the hypothesis output $h_{S}(x)$ differs from $f(x)$. We notice that the lower bound of the average risk for such classical models is determined by the feature dimension $|\mathcal{X}|$ and label dimension $|\mathcal{Y}|$ of training samples, as well as the size $t$ of training set. Whereas for TN-based models, the lower bound of the average risk is not only determined by the size of training set, but also dependent on the bond dimension $D$ and physical dimension $d$ of the TN states.  Our theorems highlight the influence of the TN model's internal structure on its performance limits.


{\it Numerical results.}--- In our previous theorems, we have analytically obtained the lower bound of the average generalization risk. To show how these theorems perform in practice, we carry out numerical simulations based on the open-source package {\tt ITensors.jl} \cite{ITensor,ITensor-r0.3} in the Julia programming language. Based on the MPS machine learning model, we consider the supervised task of learning target unitaries $U$  with the labeled training samples $\{(|\psi_j\rangle,U|\psi_j\rangle)\}$, where $|\psi_j\rangle$ denotes the normalized MPS.  We then plot the average generalization risks of the trained MPS-based learning models with respect to different qubit size $n$ of learning models, as depicted in Fig.~\ref{prac}. For example, in $n=4$, we randomly generate a 16-dimensional target unitary and a training set of MPSs, and conduct the MPS-based supervised task. By repeatedly conducting learning tasks for different target unitaries, one obtains the average generalization risk  for different size of training sets. We see from Fig.~\ref{prac} that the average error risks decrease with respect to the training set size. This is consistent with the analytical lower bound of the average risks predicted in our theorems.

\begin{figure}
\centering
\includegraphics[width=1\linewidth]{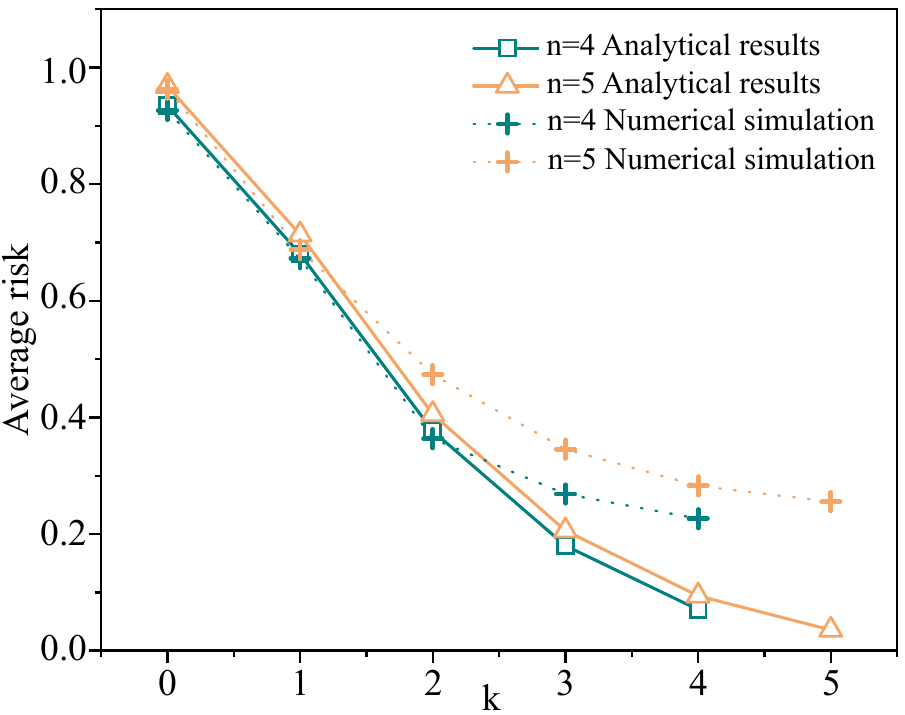}
\caption{Average  risk  of the trained MPS-based machine learning models with respect to the training set size $t_k = 2^n-2^{n-k}$, where the system qubit size $n$ varies from four to five. The physical dimension $d=2$, and the bond dimension $D=2$. The solid lines represent the analytical lower bounds of the average risk  predicted by Theorem 1. And the dotted  lines denote the average risk  of the trained MPS-based machine learning models for predicting target unitaries. }
\label{prac}
\end{figure} 

{\it Discussion.}--- Our results provide a fundamental understanding on the generalization limits of TN-based models,  extract how the performance of these specifically structured models would be limited by the NFL theorem, and analytically unveil that the lower bound of the average risk depends on both the bound and physical dimensions of TNs. Our findings would inspire further research on the learning capabilities of TN-based models with quantum computer, where TNs are employed as efficient representations of quantum circuit models. One potential direction is to incorporate the issues of practical quantum computing hardware, such as the noise and finite measurement times \cite{Wang2024Transition}, into the analytical study of generalization ability for TN-based learning models. From the perspective of experiments, future research could focus on experimentally validating NFL bounds in practical quantum computing environments. As quantum hardware continues to advance, testing these theoretical predictions on real quantum systems will be crucial to understanding how NFL constraints manifest in noisy, resource-limited settings. Such experiments would also help refine our theoretical models, potentially revealing new strategies for optimizing TN-based machine learning models for practical applications. In addition, our methods can be applied to study the generalization of deep quantum neural network \cite{Beer2020Training}, which can be efficiently trained and has been experimentally realized on a superconducting quantum processor \cite{Pan2023Deep}. We observe that the deep model can be mapped to an MPS-based model. Thus our framework can be naturally adapted to formulate the NFL for deep quantum neural networks.

In summary,  we have rigorously formulated the NFL theorems in the TN-based machine learning models. Particularly, we consider the supervised task of learning arbitrary target unitary based on the TN models, and then present the analytical lower bounds for the average risk of the models. Our results reveal the intrinsic limitations in learning arbitrary unitaries from input states encoded via TNs. The risk bounds, which depend on both the bond and physical dimensions, provide a quantitative understanding of the connections between model generalization and training set size. In addition, to validate our theoretical predictions, we numerically conduct the practical supervised learning task and show that the average generalization risks over sufficient number of learning tasks and input sets decrease with respect to the size of training dataset.  Our results offer valuable guidelines for designing more efficient models, and open  promising research directions aimed at improving the generalization capabilities of quantum-inspired TN machine learning systems.

We thank Wenjie Jiang, Zhide Lu and Weikang Li for helpful discussions. This work was supported by  the National Natural Science Foundation of China (Grant No. 12375060, No. T2225008, and No. 12075128 ), the Fundamental Research Funds for the Central Universities (Grant No. 63243070), the Shanghai Qi Zhi Institute Innovation Program SQZ202318, the Innovation Program for Quantum Science and Technology (No. 2021ZD0302203), and the Tsinghua University Dushi Program.

 \bibliography{Reference}

\clearpage
\onecolumngrid
\makeatletter
\setcounter{page}{1}
\setcounter{figure}{0}
\setcounter{equation}{0}
\setcounter{theorem}{0}

\setcounter{secnumdepth}{3}

\renewcommand{\thefigure}{S\@arabic\c@figure}
\renewcommand \theequation{S\@arabic\c@equation}
\renewcommand \thetable{S\@arabic\c@table}

\newpage

\begin{center} 
	{\large \bf Supplemental Materials: 
	No-Free-Lunch Theories for Tensor-Network Machine Learning Models}
\end{center}

\renewcommand\thesection{S\Roman{section}}
{
  \hypersetup{linkcolor=blue}
  \tableofcontents
}

\section{No-free-lunch theory for 1D tensor-network-based machine learning models}

In this section, we rigorously prove the no-free-lunch theory of 1D tensor-network based machine learning models, i.e., the matrix product state (MPS). Concretely, we consider two representative cases of supervised MPS machine learning tasks, the first case is to learn an arbitrary target unitary-embedded matrix product operator (MPO) based on the input of quantum states, and the second  case is to learn a target unitary matrix  based on the input of MPSs. We note that the second case constitutes the main part of our Theorem 1 in the main text.

For the convenience of analytical calculation, we suppose that all of the MPOs or MPSs appeared in this work are unitary embedded. 

\subsection{ Learning target matrix product operator based on the input of quantum states}
We first study the case of learning an arbitrary target unitary-embedded MPO $M$ based on the input of quantum states, where the data is encoded through unitary quantum circuits, and the variational structure is MPO. Define the training data set $\mathcal{S}$,  
\begin{equation}
\mathcal{S}=\left\{\left(|\psi_j\rangle,|\phi_j\rangle\right):|\psi_j\rangle\in \mathcal{H}_{x},|\phi_j\rangle\in \mathcal{H}_{y}\right\}_{j=1}^t,
\end{equation}
with set size $|\mathcal{S}|=t$ and $|\phi_j\rangle = M |\psi_j\rangle/\langle\psi_j|M^\dag M|\psi_j\rangle$, where $M$ denotes the variational MPO, with bond dimension.  It is worthwhile noting that $M^\dag M$ does not necessarily equal to the identity matrix. 

If the model is perfectly trained, then the trained operator $P_{\mathcal{S}}$ should result in
\begin{equation}
|\tilde{\phi_j}\rangle:=\frac{P_{\mathcal{S}}|\psi_j\rangle}{\langle\psi_j|P_{\mathcal{S}}^\dag P_{\mathcal{S}}|\psi_j\rangle} = e^{i\theta_j}\frac{M |\psi_j\rangle}{\langle\psi_j|M^\dag M|\psi_j\rangle}.
\end{equation}
One can then define the risk function for such machine learning model, as 
\begin{equation}
R_M(P_{\mathcal{S}}) = \int dx\left\| \frac{ M|x\rangle\langle x|M^\dag}{\langle x|M^\dag M|x\rangle}-\frac{P_{\mathcal{S}}|x\rangle\langle x|P_{\mathcal{S}}^\dag}{\langle x|P_{\mathcal{S}}^\dag P_{\mathcal{S}}| x\rangle}\right\|_1^2,
\end{equation} 
where  $\|A\|_1 = \frac{1}{2}\tr[\sqrt{A^\dag A}]$ denotes the trace norm of $A$. We are interested in the average risk over all training sets ${\mathcal{S}}$ and all (possible output) MPOs $M$,
\begin{equation}
\mathbb{E}_{M}\left[\mathbb{E}_{\mathcal{S}}\left[R_{M}(P_{\mathcal{S}})\right]\right].
\end{equation}
We introduce one lemma for the norm-concentration of MPS based on the Chebyschev's inequality:
\begin{equation}\label{MPS_Concentration}
\Pr_{M}\left(\left|\langle x|M^\dag M|x\rangle-1\right|\geq\epsilon\right)\leq \epsilon^2 \mathcal{O}(d^{-n}),
\end{equation}
where $d$ denotes the physical dimension of the MPO, and $n$ denote the number of qudits.

For the convenience of analytical calculations, we thus simplify the risk function $R_M(P_{\mathcal{S}})$ to be 
\begin{equation}
\begin{aligned}
R_M(P_{\mathcal{S}}) \rightarrow & \int dx\left\|  M|x\rangle\langle x|M^\dag-P_{\mathcal{S}}|x\rangle\langle x|P_{\mathcal{S}}^\dag\right\|_1^2.
\end{aligned}
\end{equation}
Suppose that the input states $|x\rangle$ form the approximate unitary-2 design, i.e., approaching to the 2-moment integral of unitary group under the Haar measure. Then the risk function has the following formula
\begin{equation}
\begin{aligned}
R_M(P_{\mathcal{S}}) \rightarrow & \int dx\left\|  M|x\rangle\langle x|M^\dag-P_{\mathcal{S}}|x\rangle\langle x|P_{\mathcal{S}}^\dag\right\|_1^2\\
\approx &  1-\int_{\rm Haar} dx \left|\langle x|M^\dag P_{\mathcal{S}}|x\rangle\right|^2\\
=&  1-\frac{N+\left|\tr(M^\dag P_{\mathcal{S}})\right|^2}{N(N+1)}.
\end{aligned}
\end{equation}
Here the term $M^\dag P_{\mathcal{S}}$ is represented by the unitary-embedded MPO.

For perfect learning strategy, where all of the  training samples in $\mathcal{S}$ are  exactly learned, but up to an overall phase, one has
\begin{equation}
\langle\psi_j|M^\dag P_{\mathcal{S}}|\psi_j\rangle=e^{i\theta_j}.
\end{equation}
We consider three cases of training states \cite{Sharma2022Reformulation}:
\begin{itemize}
\item[(a)] States in $\mathcal{S}$ are orthonormal. $W=M^\dag P_{\mathcal{S}} =e^{i\theta_1}\oplus e^{i\theta_2} \cdots \oplus  e^{i\theta_t}\oplus Y$. $Y$ is the $(d^n-t)$-dimensional matrix composed by MPOs. $\left|\tr(M^\dag P_{\mathcal{S}})\right|^2=\left|\sum_{j=1}^t e^{i\theta_j}+\tr[Y]\right|^2$.
\item[(b)] States in $\mathcal{S}$ are non-orthonormal, but linear independent. $W=M^\dag P_{\mathcal{S}} =e^{i\theta I_{t}} \oplus Y$. $Y$ is the $(d^n-t)$-dimensional matrix composed by MPOs. $\left|\tr(M^\dag P_{\mathcal{S}})\right|^2=\left|t e^{i\theta}+\tr[Y]\right|^2$.
\item[(c)] States in $\mathcal{S}$ are linear dependent. $W=M^\dag P_{\mathcal{S}} =e^{i\theta I_{t'}} \oplus Y$, $t'$ denotes the linear independent bases of the training set.  $Y$ is the $(d^n-t')$-dimensional matrix composed by MPOs.
\end{itemize}

In the case (a), the matrix $Y$ is $(d^n-t)$-dimensional. For the convenience of analytical calculation, we suppose that $Y$ can be represented by a new unitary-embedded MPO with bond dimension $D'$,  physical dimension $d'$, and $n'$ physical sites, as long as it satisfies $d'^{n'} = d^n-t$.  We can always make such assumption, since the integral $\int dY \tr[Y]\tr[Y^*]=1$ holds for random unitary-embedded  MPO $Y$ with arbitrary bond dimension  and physical dimension, see the following proof.

Now we calculate the average risk over all possible MPO $M$, 
\begin{equation}\label{Risk_aver_M}
\begin{aligned}
 \mathbb{E}_{M}\left[R_{M}(P_{\mathcal{S}})\right] =& 1-\frac{d^n+\int dM |\tr(M^\dag P_{\mathcal{S}})|^2}{d^n(d^n+1)}\\
=& 1-\frac{d^n+ \int dY\left|\sum_{j=1}^t e^{i\theta_j}+\tr[Y]\right|^2}{d^n(d^n+1)}\\
=& 1-\frac{d^n+  \int dY\left[\left|\sum_{j=1}^t e^{i\theta_j}\right|^2+\left|\tr[Y]\right|^2+\sum_{j=1}^t e^{i\theta_j}\tr[Y^\dag]+\sum_{j=1}^t e^{-i\theta_j}\tr[Y]\right]}{d^n(d^n+1)}\\
\geq & 1-\frac{d^n+t^2+\int dY \tr[Y]\tr[Y^*] }{d^n(d^n+1)} = 1-\frac{d^n+t^2+1}{d^n(d^n+1)}.
\end{aligned}
\end{equation}

Now we prove $\int dY \tr[Y]\tr[Y^*]=1$ for random unitary-embedded  MPO $Y$in periodic boundary condition.
\begin{proof}
The unitary embedded MPO $Y$ takes the following form 
\begin{equation}\label{YMPO_PIC}
Y = \ipic{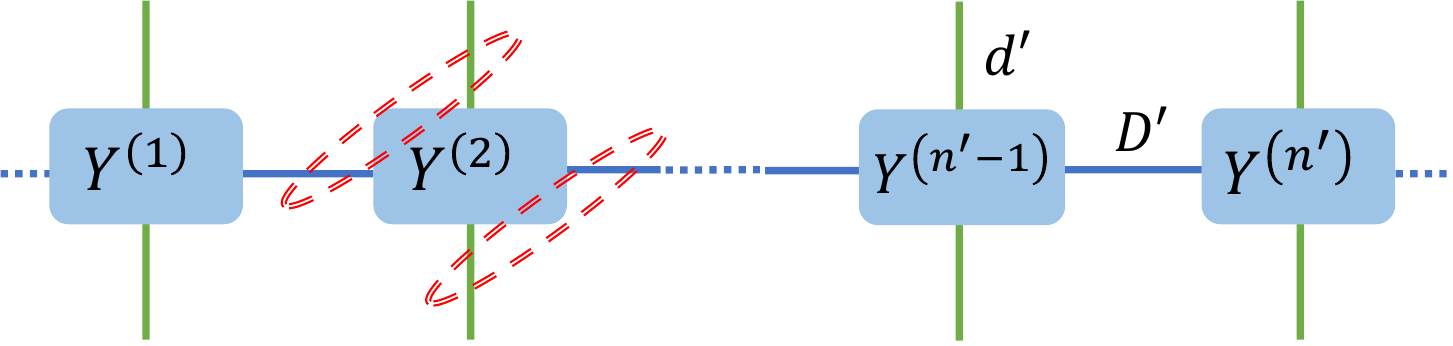}{0.3} ,
\end{equation}
where $Y^{(k)}$  locates on the $k$-th physical site. Through unitary embedding \cite{Haferkamp2021Emergent,Liu2022Presence}, one can map the 4-rank tensor $Y^{(k)}$ to a 2-rank unitary matrix in the unitary group $SU(D'd')$. Thus, the four dangling legs connected to $Y^{(k)}$ are divided into two parities, with each party (circled by the red lines in Eq.~(\ref{YMPO_PIC}) ) hosts $D'd'$-dimension. 

The random tensor $Y^{(k)}$ forms the approximate unitary 1-design, and the 1-moment integral of $Y^{(k)}$ reads
\begin{equation}
\int_{\rm Haar} dY^{(k)} Y^{(k)}_{ij,lm}Y^{*(k)}_{i'j',l'm'}=\int_{\rm Haar} dY^{(k)}\ipic{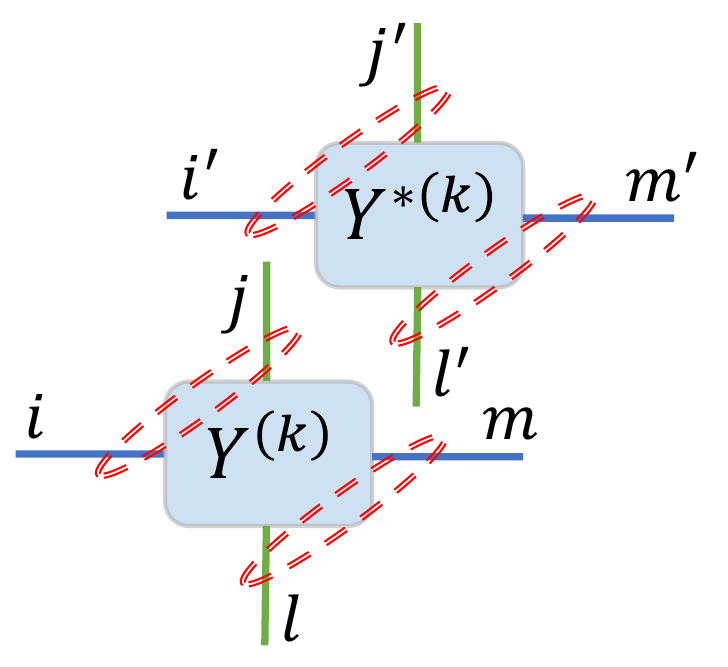}{0.25} =\frac{1}{D'd'} \ipic{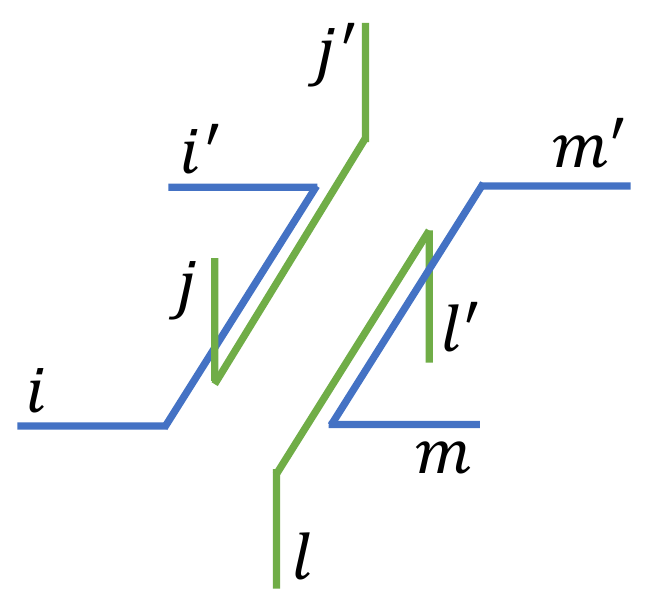}{0.25}=\frac{1}{D'd'} \delta_{ii'}\delta_{jj'}\delta_{ll'}\delta_{mm'}.
\end{equation}
Then the 1-moment integral of $Y$ in Eq.~(\ref{Risk_aver_M}) can be expressed as
\begin{equation}
\begin{aligned}
&\int dY \tr[Y]\tr[Y^*]\\
=&\int dY\ipic{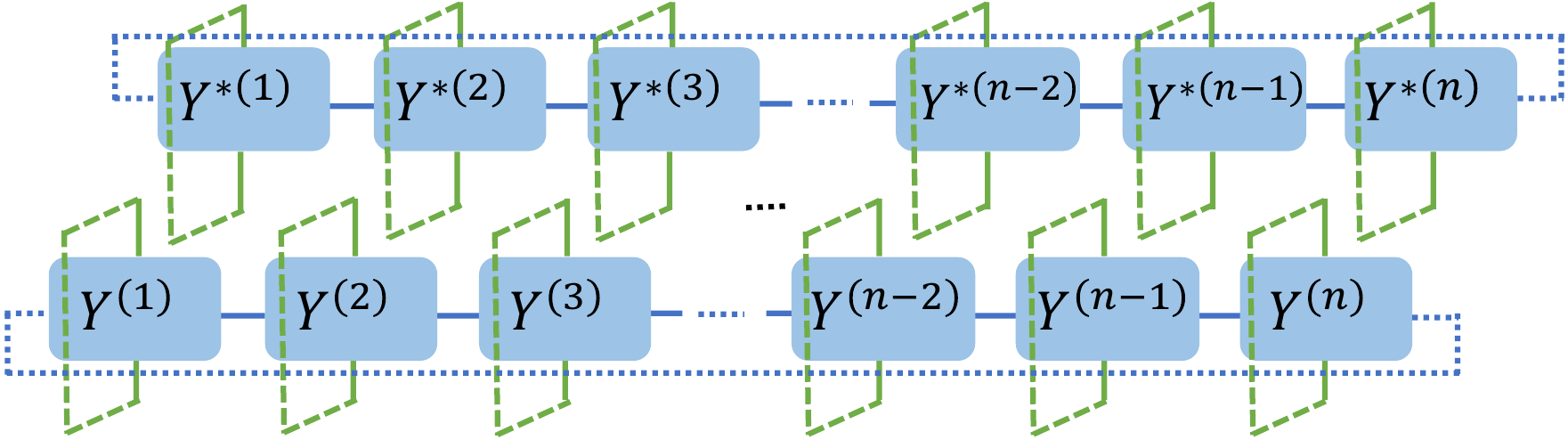}{0.25} \\
=&\frac{1}{(D'd')^n} \ipic{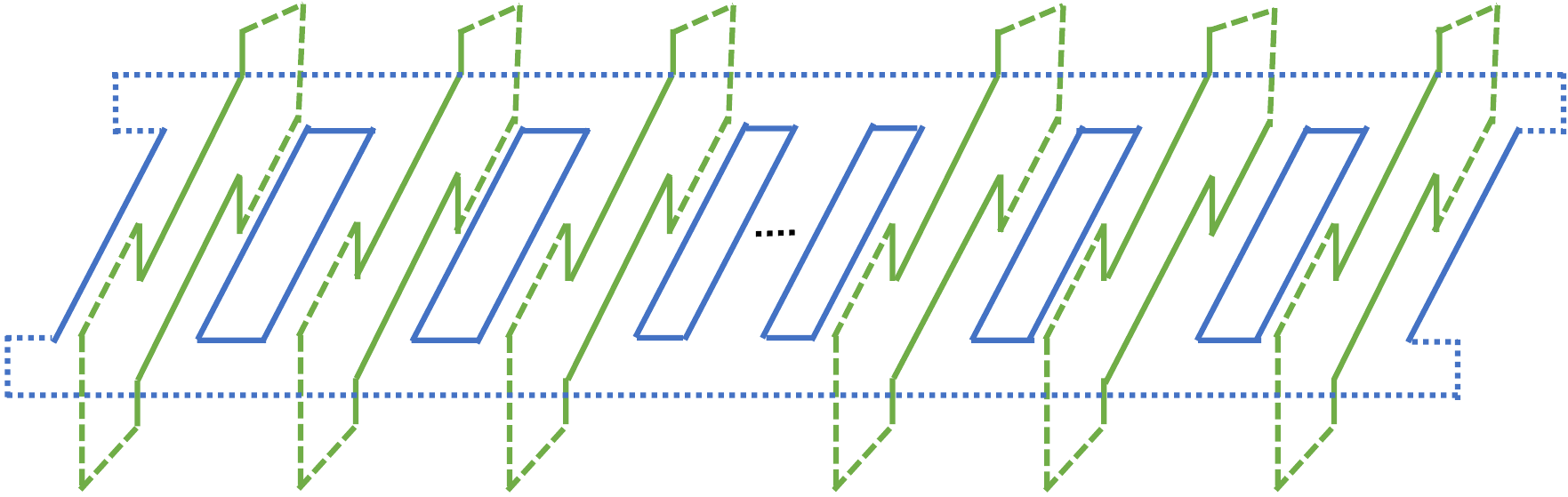}{0.25}\\
=&1,
\end{aligned}
\end{equation}
where each green loop contributes to one $d'$-factor, and each blue loop contributes to one $D'$-factor.
We see from the above equation that the value of integral $\int dY \tr[Y]\tr[Y^*]$ is independent of the concrete values of bond dimension $D'$ and physical dimension $d'$. Therefore our assumption on representing $Y$ by a new random MPO is reasonable. This completes our proof. 
\end{proof}
Since the result on  $\mathbb{E}_{M}\left[R_{M}(P_{\mathcal{S}})\right]$ is independent of the choice of training set, thus the average risk over all training sets ${\mathcal{S}}$ and all (possible output) MPOs $M$ has the same formula as $\mathbb{E}_{M}\left[R_{M}(P_{\mathcal{S}})\right]$:

\begin{equation}
\mathbb{E}_{M}\left[\mathbb{E}_{\mathcal{S}}\left[R_{M}(P_{\mathcal{S}})\right]\right]=\mathbb{E}_{M}\left[R_{M}(P_{\mathcal{S}})\right]\geq 1-\frac{d^n+t^2+1}{d^n(d^n+1)}.
\end{equation}

Similar calculations also apply to the cases (b) and (c). 

Thus for arbitrary cases of training states, we obtain that 
\begin{itemize}
\item Case (a): States in $\mathcal{S}$ are orthonormal. $W=M^\dag P_{\mathcal{S}} =e^{i\theta_1}\oplus e^{i\theta_2} \cdots \oplus  e^{i\theta_t}\oplus Y$. $Y$ is the $(d^n-t)$-dimensional matrix composed by MPOs. $\left|\tr(M^\dag P_{\mathcal{S}})\right|^2=\left|\sum_{j=1}^t e^{i\theta_j}+\tr[Y]\right|^2$.
\begin{equation}
\mathbb{E}_{M}\left[\mathbb{E}_{\mathcal{S}}\left[R_{M}(P_{\mathcal{S}})\right]\right]=\mathbb{E}_{M}\left[R_{M}(P_{\mathcal{S}})\right]\geq 1-\frac{d^n+t^2+1}{d^n(d^n+1)}.
\end{equation}
\item Case (b): States in $\mathcal{S}$ are non-orthonormal, but linear independent. $W=M^\dag P_{\mathcal{S}} =e^{i\theta I_{t}} \oplus Y$. $Y$ is the $(d^n-t)$-dimensional matrix composed by MPOs. $\left|\tr(M^\dag P_{\mathcal{S}})\right|^2=\left|t e^{i\theta}+\tr[Y]\right|^2$.
\begin{equation}
\mathbb{E}_{M}\left[\mathbb{E}_{\mathcal{S}}\left[R_{M}(P_{\mathcal{S}})\right]\right]=\mathbb{E}_{M}\left[R_{M}(P_{\mathcal{S}})\right]= 1-\frac{d^n+t^2+1}{d^n(d^n+1)}.
\end{equation}
\item Case (c): States in $\mathcal{S}$ are linear dependent. $W=M^\dag P_{\mathcal{S}} =e^{i\theta I_{t'}} \oplus Y$, $t'$ denotes the linear independent bases of the training set.  $Y$ is the $(d^n-t')$-dimensional matrix composed by MPOs.
\begin{equation}
\mathbb{E}_{M}\left[\mathbb{E}_{\mathcal{S}}\left[R_{M}(P_{\mathcal{S}})\right]\right]=\mathbb{E}_{M}\left[R_{M}(P_{\mathcal{S}})\right]= 1-\frac{d^n+t'^2+1}{d^n(d^n+1)}.
\end{equation}
\end{itemize}

\subsection{Learning target quantum unitary operator based on the input of matrix product states}
Now we consider the case of learning target quantum unitary operator based on the input of MPSs, where the training states are encoded into the unitary-embedded MPSs, and the variational structures are unitary quantum circuits. 

Owing to the norm-concentration of MPSs in Eq.~(\ref{MPS_Concentration}), one thus simplifies the risk function $R_M(P_{\mathcal{S}})$ to be 
\begin{equation}\label{risk_x_mps_vs_vc}
\begin{aligned}
R_M(P_{\mathcal{S}}) \rightarrow &\frac{1}{4} \int dx\left\|  M|x\rangle\langle x|M^\dag-P_{\mathcal{S}}|x\rangle\langle x|P_{\mathcal{S}}^\dag\right\|_1^2\\
\approx &  1-\int dx \left|\langle x|M^\dag P_{\mathcal{S}}|x\rangle\right|^2.
\end{aligned}
\end{equation}
Here $|x\rangle$ is the random unitary-embedded MPS,
\begin{equation}
|x\rangle = \ipic{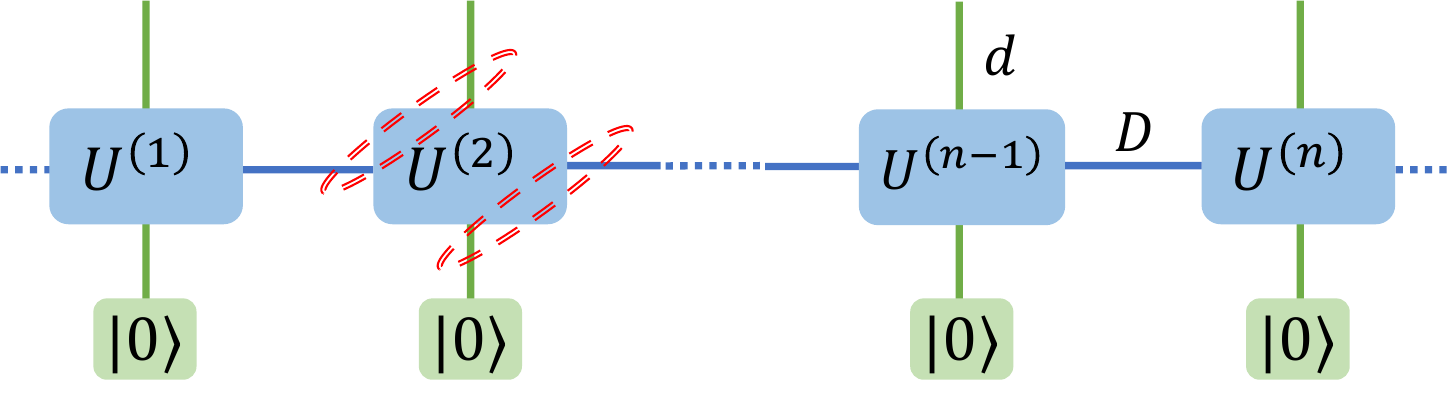}{0.25},
\end{equation}
where $U^{(k)}$ denotes the $Dd\times Dd$ unitary matrix locating on the $k$-th site.
In Eq.~(\ref{risk_x_mps_vs_vc}), One needs to calculate the 2-moment integral of the random MPS $|x\rangle$. We represent the  variational circuit by the unitary  $W = M^\dag P_{\mathcal{S}}$. Following the above statements on the training states, we have the three following cases
\begin{itemize}
\item[(a)] States in $\mathcal{S}$ are orthonormal. $W=M^\dag P_{\mathcal{S}} =e^{i\theta_1}\oplus e^{i\theta_2} \cdots \oplus  e^{i\theta_t}\oplus Y$. $Y$ is the $(d^n-t)$-dimensional unitary matrix. 
\item[(b)] States in $\mathcal{S}$ are non-orthonormal, but linear independent. $W=M^\dag P_{\mathcal{S}} =e^{i\theta I_{t}} \oplus Y$. $Y$ is the $(d^n-t)$-dimensional unitary matrix. 
\item[(c)] States in $\mathcal{S}$ are linear dependent. $W=M^\dag P_{\mathcal{S}} =e^{i\theta I_{t'}} \oplus Y$, $t'$ denotes the linear independent bases of the training set.  $Y$ is the $(d^n-t')$-dimensional unitary matrix.
\end{itemize} 

Without loss of generality, we focus on the case (a). To calculate the risk function,  one needs to calculate the 2-moment integral of random unitary-embedded MPS $|x\rangle$ 
\begin{equation}\label{x_2_moment_integral}
\begin{aligned}
\int dx \left|\langle x|M^\dag P_{\mathcal{S}}|x\rangle\right|^2=\ipic{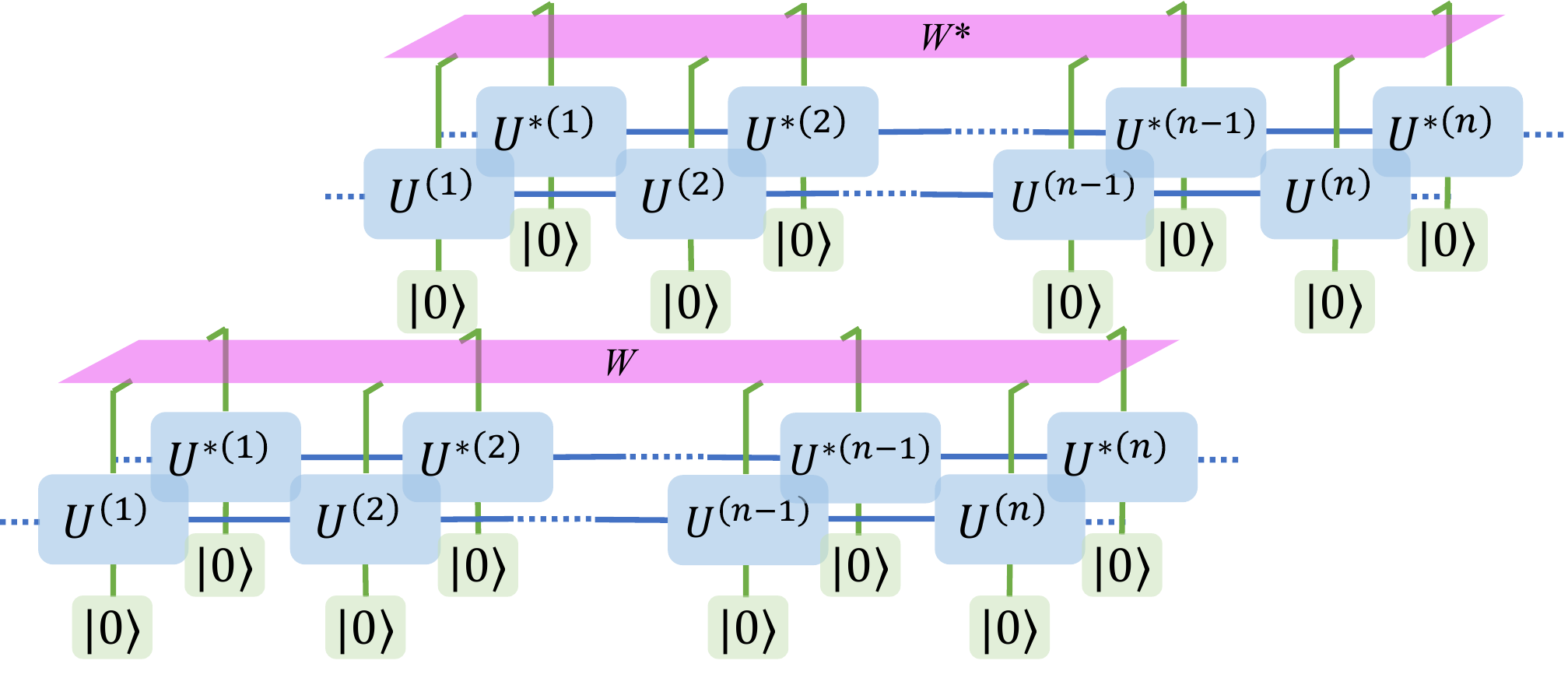}{0.25},
\end{aligned}
\end{equation}
where  all of the $U^{(k)}$'s are supposed to distribute randomly and form the approximate unitary 2-design. In previous works \cite{Haferkamp2021Emergent,Liu2022Presence}, it has been noted that the calculation of Eq.~(\ref{x_2_moment_integral}) is equivalent to calculating the partition function of 1D classical Ising model. For each local tensor $U^{(k)}$, one has
\begin{equation}
\begin{aligned}
\int d U^{(k)} \ipic{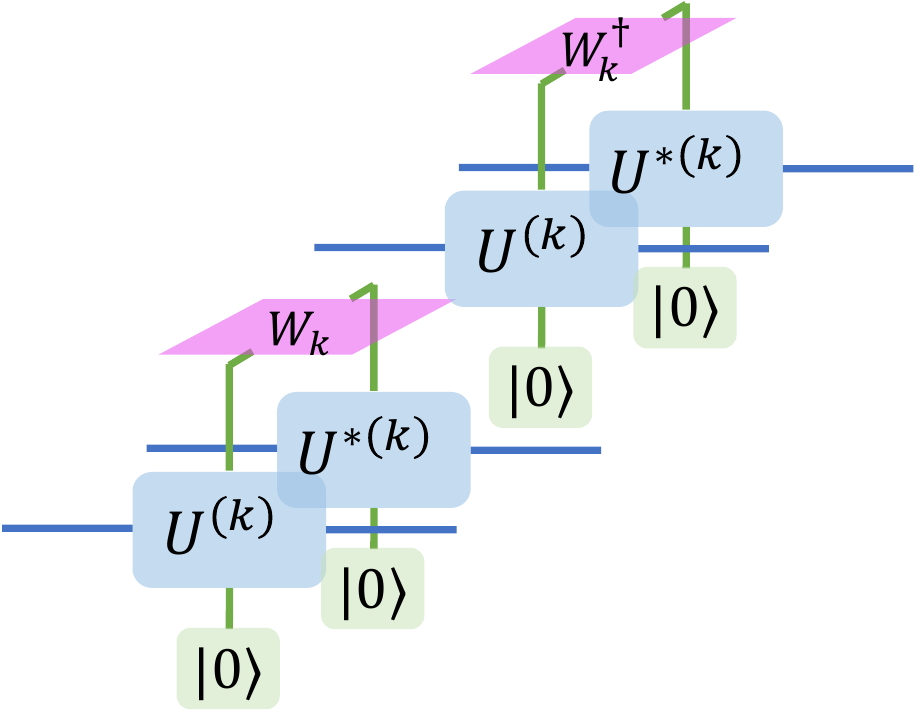}{0.25} = & \frac{1}{(Dd)^2-1}\left(\ipic{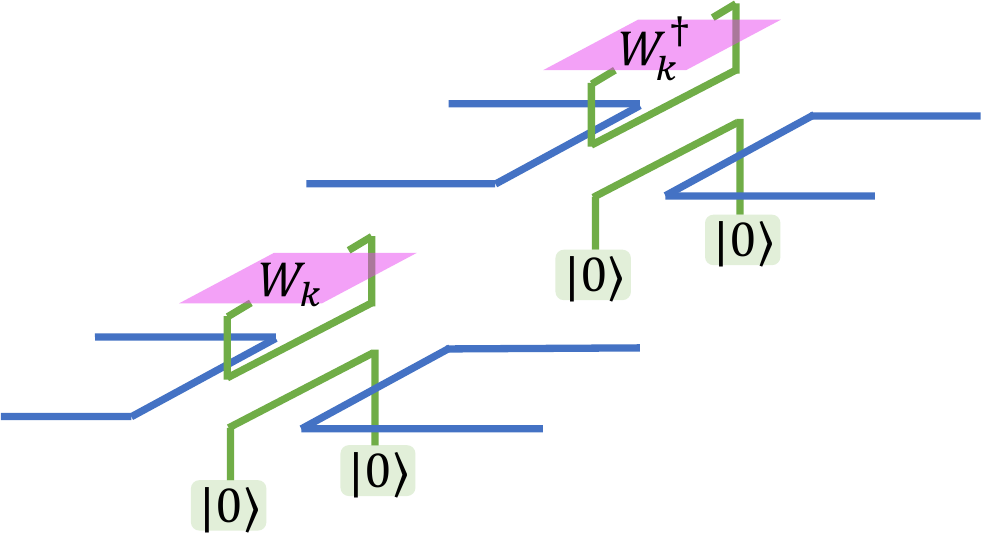}{0.25}+\ipic{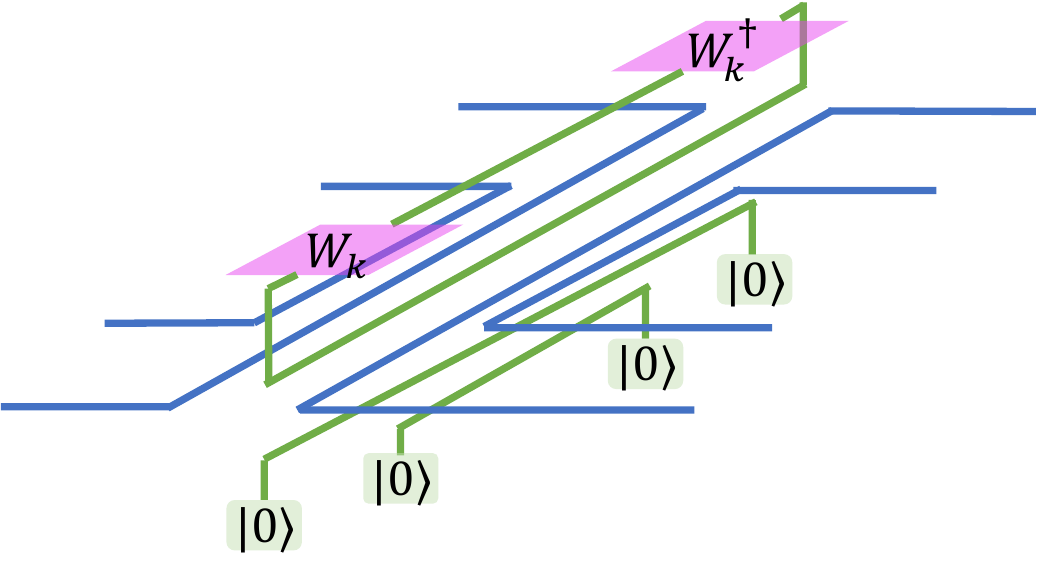}{0.25}\right)\\
&-\frac{1}{(Dd)^3-Dd}\left(\ipic{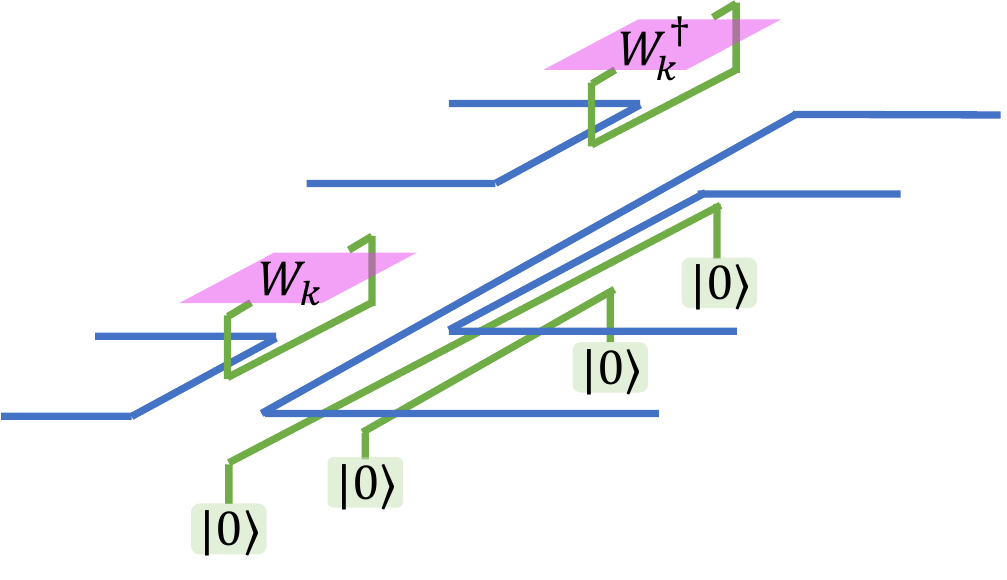}{0.25}+\ipic{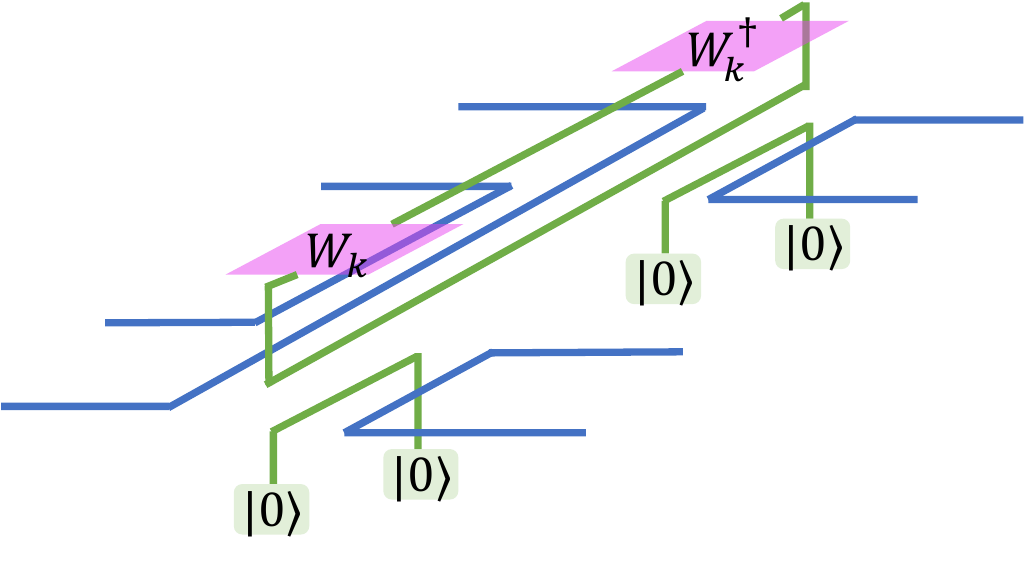}{0.25}\right)\\
=&\frac{1}{(Dd)^2-1}\left(S^{(k)} \tr[W_k]\tr[W_k^\dag] S^{(k)}+A^{(k)} \tr(W_kW_k^\dag) A^{(k)}\right)\\
&-\frac{1}{(Dd)^3-Dd}\left(S^{(k)} \tr[W_k]\tr[W_k^\dag] A^{(k)}+A^{(k)} \tr[W_kW_k^\dag] S^{(k)}\right) \\
= &\frac{1}{(Dd)^2-1}\left(S^{(k)} \tr[W_k\otimes W_k^\dag] S^{(k)}+A^{(k)}  \tr[{\rm SWAP}_{k}W_k\otimes W_k^\dag] A^{(k)}\right)\\
&-\frac{1}{(Dd)^3-Dd}\left(S^{(k)} \tr[W_k\otimes W_k^\dag] A^{(k)}+A^{(k)} \tr[{\rm SWAP}_{k}W_k\otimes W_k^\dag] S^{(k)}\right) ,
\end{aligned}
\end{equation}
where ${\rm SWAP}_{k}$ denotes the swapping operation on the $k$-th site of $W^\dag$ and $W$, $S^{(k)}$ and $A^{(k)}$ denote the symmetric and anti-symmetric connections of the bond indices of $U^{(k)}$  \cite{Liu2022Presence}, 
\begin{equation}
S = \ipic{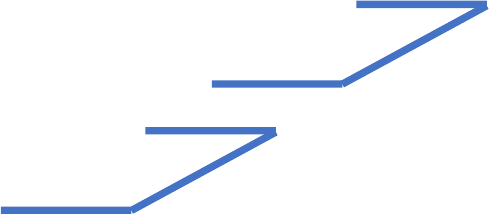}{0.25},\quad A = \ipic{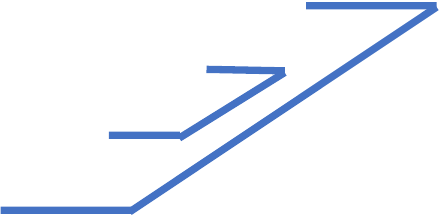}{0.25},
\end{equation}
and the term $\langle0|$---$|0\rangle=1$.
Here  the notation $W_k$ just denotes how the local indices of $W_{i_1i_2...i_k...i_n,j_1j_2...j_k...j_n}$ contract on the $k$-th physical site, rather than the reduced matrix under partial trace. 

One can embed the calculation of 2-moment integral of random MPS $|x\rangle$ in Eq.~(\ref{x_2_moment_integral})  into the calculation of partition function. Here we utilize the transfer matrix method to calculate the partition function,
\begin{eqnarray}
\ipic{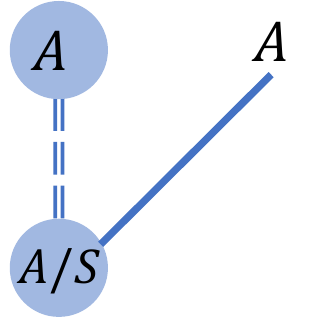}{0.25} = \frac{D^2}{(Dd)^2-1}-\frac{D}{(Dd)^3-Dd}=\frac{D^2d-1}{D^2d^3-d},\\
\ipic{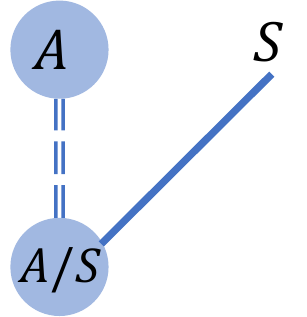}{0.25} = \frac{D}{(Dd)^2-1}-\frac{D^2}{(Dd)^3-Dd}=\frac{Dd-D}{D^2d^3-d},\\
\ipic{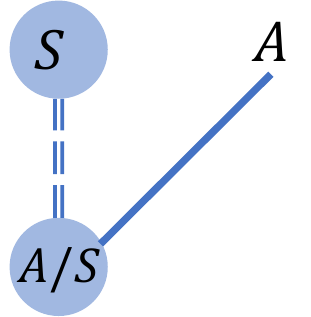}{0.25} =-\frac{D^2}{(Dd)^3-Dd}+ \frac{D}{(Dd)^2-1}=\frac{Dd-D}{D^2d^3-d},\\
\ipic{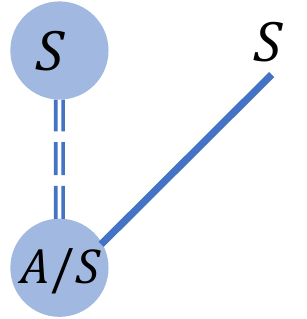}{0.25} =-\frac{D}{(Dd)^3-Dd}+ \frac{D^2}{(Dd)^2-1}=\frac{D^2d-1}{D^2d^3-d}.
\end{eqnarray} 
We thus obtain the transfer matrix 
\begin{equation}
T=\left[\begin{matrix}
&T_{AA}, & T_{AS} \\
&T_{SA}, & T_{SS}
\end{matrix}\right]=\left[\begin{matrix}
&\frac{D^2d-1}{D^2d^3-d}, & \frac{Dd-D}{D^2d^3-d}\\
&\frac{Dd-D}{D^2d^3-d}, & \frac{D^2d-1}{D^2d^3-d}
\end{matrix}\right].
\end{equation}

Not that $\tr[{\rm SWAP } A\otimes B]=\tr[AB]$, thus the Eq.~(\ref{x_2_moment_integral}) takes the form
\begin{equation}
\begin{aligned}
\int dx \left|\langle x|M^\dag P_{\mathcal{S}}|x\rangle\right|^2 = \int dx \left|\langle x|W|x\rangle\right|^2 = \tr\left[\mathcal{F}\left(\left\{{\rm SWAP}_{i}| i=1,2,...n\right\}\right)W\otimes W^\dag\right].
\end{aligned}
\end{equation}
In the periodic boundary condition of MPS, the term 
\begin{equation}
\begin{aligned}
&\mathcal{F}\left(\left\{{\rm SWAP}_{i}| i=1,2,...n\right\}\right) \\
=& \bar{\tr}\left[\left[\begin{matrix}
&{\rm SWAP}_{1}, & 0 \\
&0, & \mathbb{I}
\end{matrix}\right]\left[\begin{matrix}
&T_{AA}, & T_{AS} \\
&T_{SA}, & T_{SS}
\end{matrix}\right]\left[\begin{matrix}
&{\rm SWAP}_{2}, & 0 \\
&0, & \mathbb{I}
\end{matrix}\right]\left[\begin{matrix}
&T_{AA}, & T_{AS} \\
&T_{SA}, & T_{SS}
\end{matrix}\right]...\left[\begin{matrix}
&{\rm SWAP}_{n}, & 0 \\
&0, & \mathbb{I}
\end{matrix}\right]\left[\begin{matrix}
&T_{AA}, & T_{AS} \\
&T_{SA}, & T_{SS}
\end{matrix}\right]\right],
\end{aligned}
\end{equation}
Here the trace operation $\bar{\tr}$ only applies on the space of transfer matrix, rather than the swapping space. Graphically, one has
\begin{equation}
\tr\left[\mathcal{F}\left(\left\{{\rm SWAP}_{i}| i=1,2,...n\right\}\right)W\otimes W^\dag \right]= \ipic{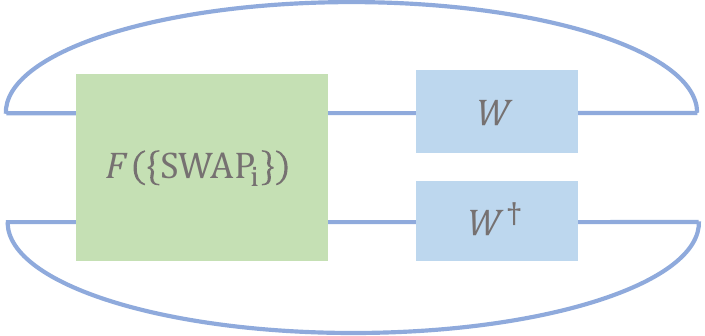}{0.5}
\end{equation}

Now we calculate the average of risk function over all (possible output) unitaries $M$.
\begin{equation}
\begin{aligned}
\mathbb{E}_{M}\left[R_{M}(P_{\mathcal{S}})\right] =1- \int dW \tr\left[\mathcal{F}\left(\left\{{\rm SWAP}_{i}| i=1,2,...n\right\}\right)W\otimes W^\dag\right].
\end{aligned}
\end{equation}
 We consider two special cases of the training set size $t$, the empty training set with $t=0$, and the full training set with $t=2^n$.
 
 For the case of empty training set, $t=0$, and the 1-moment integral over the Haar ensemble of $W$ 
 \begin{equation}
\begin{aligned}
&\mathbb{E}_{M}\left[R_{M}(P_{\mathcal{S}})\right] \\
=&1- \int_{\rm Haar} dW \tr\left[\mathcal{F}\left(\left\{{\rm SWAP}_{i}| i=1,2,...n\right\}\right)W\otimes W^\dag\right]\\
=&1- \frac{1}{d^n}\ipic{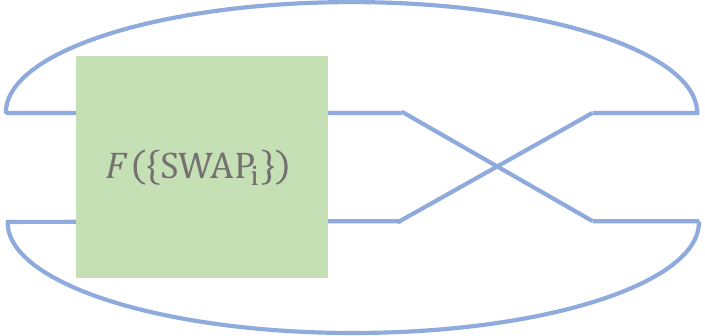}{0.5}\\
=&1- \frac{1}{d^n}\tr\left[\mathcal{F}\left(\left\{{\rm SWAP}_{i}| i=1,2,...n\right\}\right)\otimes_{i=1}^n{\rm SWAP}_{i}\right]\\
=&1- \frac{1}{d^n}\tr\left[\left[\begin{matrix}
&d^2, & 0 \\
&0, & d
\end{matrix}\right]\left[\begin{matrix}
&T_{AA}, & T_{AS} \\
&T_{SA}, & T_{SS}
\end{matrix}\right]\left[\begin{matrix}
&d^2, & 0 \\
&0, & d
\end{matrix}\right]\left[\begin{matrix}
&T_{AA}, & T_{AS} \\
&T_{SA}, & T_{SS}
\end{matrix}\right]...\left[\begin{matrix}
&d^2, & 0 \\
&0, & d
\end{matrix}\right]\left[\begin{matrix}
&T_{AA}, & T_{AS} \\
&T_{SA}, & T_{SS}
\end{matrix}\right]\right]\\
=& 1- \tr\left[\left(\left[\begin{matrix}
&d, & 0 \\
&0, & 1
\end{matrix}\right]\left[\begin{matrix}
&T_{AA}, & T_{AS} \\
&T_{SA}, & T_{SS}
\end{matrix}\right]\right)^n\right]\\
=& 1-\tr\left[\left[\begin{matrix}&\frac{D^2d-1}{D^2d^2-1}, & \frac{Dd-D}{D^2d^2-1}\\
&\frac{Dd-D}{D^2d^3-d}, & \frac{D^2d-1}{D^2d^3-d}\end{matrix}\right]^n\right]\\
=&1-\left[\frac{1}{d^n}+\left(\frac{D^2-1}{D^2d^2-1}\right)^n\right]\\
\approx  & 1.
\end{aligned}
\end{equation}

For the case of full training set, $t=d^n$, then  the learned unitary $M$ obeys   $M^\dag P_{\mathcal{S}}  =W =e^{i\theta_1}\oplus e^{i\theta_2} \cdots \oplus  e^{i\theta_t}\oplus e^{i\theta_{2^n}}$, thus
 \begin{equation}
\begin{aligned}
&\mathbb{E}_{M}\left[R_{M}(P_{\mathcal{S}})\right] \\
=&1-  \tr\left[\mathcal{F}\left(\left\{{\rm SWAP}_{i}| i=1,2,...n\right\}\right)W\otimes W^\dag\right]\\
\geq&1-\left( \tr\left[\mathcal{F}\left(\left\{{\rm SWAP}_{i}| i=1,2,...n\right\}\right)(WW^\dag)\otimes I\right]\tr\left[\mathcal{F}\left(\left\{{\rm SWAP}_{i}| i=1,2,...n\right\}\right)I \otimes (W^\dag W)\right]\right)^{1/2}\\
=&1- \ipic{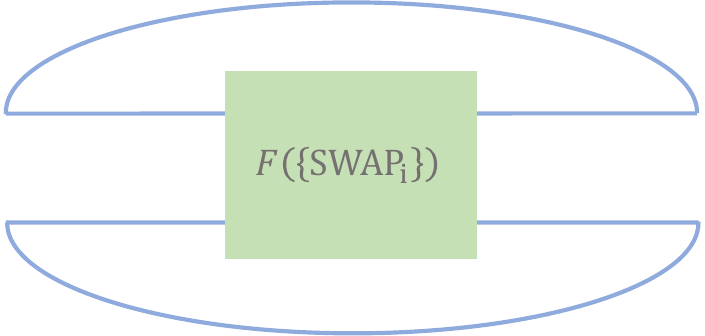}{0.5}\\
=&1- \tr\left[\mathcal{F}\left(\left\{{\rm SWAP}_{i}| i=1,2,...n\right\}\right)\right]\\
=&1- \tr\left[\left[\begin{matrix}
&d, & 0 \\
&0, & d^2
\end{matrix}\right]\left[\begin{matrix}
&T_{AA}, & T_{AS} \\
&T_{SA}, & T_{SS}
\end{matrix}\right]\left[\begin{matrix}
&d, & 0 \\
&0, & d^2
\end{matrix}\right]\left[\begin{matrix}
&T_{AA}, & T_{AS} \\
&T_{SA}, & T_{SS}
\end{matrix}\right]...\left[\begin{matrix}
&d, & 0 \\
&0, & d^2
\end{matrix}\right]\left[\begin{matrix}
&T_{AA}, & T_{AS} \\
&T_{SA}, & T_{SS}
\end{matrix}\right]\right]\\
=& 1- \tr\left[\left(\left[\begin{matrix}
&d, & 0 \\
&0, & d^2
\end{matrix}\right]\left[\begin{matrix}
&T_{AA}, & T_{AS} \\
&T_{SA}, & T_{SS}
\end{matrix}\right]\right)^n\right]\\
=& 1-\tr\left[\left[\begin{matrix}&\frac{D^2d-1}{D^2d^2-1}, & \frac{Dd-D}{D^2d^2-1}\\
&\frac{Dd^2-Dd}{D^2d^2-1}, & \frac{D^2d^2-d}{D^2d^2-1}\end{matrix}\right]^n\right]\\
=&1-\left[1+\left(\frac{D^2d-d}{D^2d^2-1}\right)^n\right]\\
\approx  & 0.
\end{aligned}
\end{equation}
 In the above proof, we utilize the equality  $\tr\left[\mathcal{F}\left(\left\{{\rm SWAP}_{i}| i=1,2,...n\right\}\right)W\otimes W^\dag\right]= \tr\left[\mathcal{F}\left(\left\{{\rm SWAP}_{i}| i=1,2,...n\right\}\right)W^\dag \otimes W\right]$, and  the matrix trace inequality
\begin{equation}
\left|\tr[F A^\dag B]\right|^2\leq\tr[F A^\dag A]\tr[F B^\dag B],
\end{equation}
where $F$ is positive semi-definite with $\tr(F)=1$. 

For the convenience of analytical calculations,  we consider the case of $t_k=d^n-d^{n-k}$ training samples, and suppose that the learned unitary $P_{\mathcal{S}}$ obeys   $M^\dag P_{\mathcal{S}} =W=e^{i\theta I_{t_k}}\oplus Y$, such that $Y$ is a $d^{n-k}\times d^{n-k}$ unitary matrix. For simplicity, one suppose that   the unitary $Y$ locates in a $(n-k)$-qudit subsystem. Thus the average risk function over $M$ reads
 \begin{equation}\label{Supp_Y_tk}
\begin{aligned}
&\mathbb{E}_{M}\left[R_{M}(P_{\mathcal{S}})\right] \\
=&1-  \int_{\rm Haar} dY \tr\left[\mathcal{F}\left(\left\{{\rm SWAP}_{i}| i=1,2,...n\right\}\right)(e^{i\theta I_{t_k}}\oplus Y)\otimes (e^{-i\theta I_{t_k}}\oplus Y^\dag)\right]\\
=&1-  \int_{\rm Haar} dY \tr\left[\mathcal{F}\left(\left\{{\rm SWAP}_{i}| i=1,2,...n\right\}\right)
(e^{i\theta}(I_d^{\otimes n}-\Sigma^{\otimes k}\otimes I_d^{\otimes n-k}) +\Sigma^{\otimes k}\otimes Y)\right.\\
&\qquad\qquad\qquad \quad \left.\otimes (e^{-i\theta}(I_d^{\otimes n}-\Sigma^{\otimes k}\otimes I_d^{\otimes n-k}) +\Sigma^{\otimes k}\otimes Y^\dag)\right]\\
=&1-  \int_{\rm Haar} dY \tr\left[\mathcal{F}\left(\left\{{\rm SWAP}_{i}| i=1,2,...n\right\}\right)
(I_d^{\otimes n}-\Sigma^{\otimes k}\otimes I_d^{\otimes n-k})^{\otimes2} +\Sigma^{\otimes k}\otimes Y\otimes\Sigma^{\otimes k}\otimes Y^\dag\right],
\end{aligned}
\end{equation}
 where $\Sigma = {\rm Diag}_d\{0,0,....,1\}$ denotes a $d\times d$ matrix.
 
 We first calculate the term 
 \begin{equation}
 \begin{aligned}
 &\tr\left[\mathcal{F}\left(\left\{{\rm SWAP}_{i}| i=1,2,...n\right\}\right)
(I_d^{\otimes n}-\Sigma^{\otimes k}\otimes I_d^{\otimes n-k})^{\otimes2}\right]\\
=&\tr\left[\mathcal{F}\left(\left\{{\rm SWAP}_{i}| i=1,2,...n\right\}\right)
((I_d^{\otimes n})^{\otimes2}-I_d^{\otimes n}\otimes\Sigma^{\otimes k}\otimes I_d^{\otimes n-k}-\Sigma^{\otimes k}\otimes I_d^{\otimes n-k}\otimes I_d^{\otimes n}+(\Sigma^{\otimes k}\otimes I_d^{\otimes n-k})^{\otimes2})\right]
\end{aligned}
\end{equation}
For arbitrary local term, one has
\begin{eqnarray}
&&\tr\left[{\rm SWAP}_i I_{d,i}\otimes I_{d,i}\right] = d,\qquad \tr\left[ I_{d,i}\otimes I_{d,i}\right] = d^2;\\
&&\tr\left[{\rm SWAP}_i I_{d,i}\otimes \Sigma_i\right] = 1,\qquad \tr\left[ I_{d,i}\otimes \Sigma_i\right] = d;\\
&&\tr\left[{\rm SWAP}_i \Sigma_i\otimes I_{d,i}\right] = 1,\qquad \tr\left[ \Sigma_i\otimes I_{d,i}\right] = d;\\
&&\tr\left[{\rm SWAP}_i \Sigma_i\otimes \Sigma_i\right] = 1,\qquad \tr\left[ \Sigma_i\otimes \Sigma_i\right] = 1.
\end{eqnarray}
Thus
\begin{equation}
\tr\left[\mathcal{F}\left(\left\{{\rm SWAP}_{i}| i=1,2,...n\right\}\right)
(I_d^{\otimes n})^{\otimes2}\right] =\tr\left[\mathcal{F}\left(\left\{{\rm SWAP}_{i}| i=1,2,...n\right\}\right)\right] = 1+\left(\frac{D^2d-d}{D^2d^2-1}\right)^n,
\end{equation}
\begin{equation}
\begin{aligned}
\tr\left[\mathcal{F}\left(\left\{{\rm SWAP}_{i}| i=1,2,...n\right\}\right)
(I_d^{\otimes n}\otimes(\Sigma^{\otimes k}\otimes I_d^{\otimes n-k}))\right]=&\tr\left[\left(\left[\begin{matrix}
&1, & 0 \\
&0, & d
\end{matrix}\right]T\right)^{k}\left(\left[\begin{matrix}
&d, & 0 \\
&0, & d^2
\end{matrix}\right]T\right)^{n-k}\right]\\
=& \frac{1}{d^k}\left(1+\left(\frac{D^2d-d}{D^2d^2-1}\right)^n\right),
\end{aligned}
\end{equation}
\begin{equation}
\begin{aligned}
\tr\left[\mathcal{F}\left(\left\{{\rm SWAP}_{i}| i=1,2,...n\right\}\right)(
(\Sigma^{\otimes k}\otimes I_d^{\otimes n-k})\otimes I_d^{\otimes n})\right]=&\tr\left[\left(\left[\begin{matrix}
&1, & 0 \\
&0, & d
\end{matrix}\right]T\right)^{k}\left(\left[\begin{matrix}
&d, & 0 \\
&0, & d^2
\end{matrix}\right]T\right)^{n-k}\right]\\
=& \frac{1}{d^k}\left(1+\left(\frac{D^2d-d}{D^2d^2-1}\right)^n\right),
\end{aligned}
\end{equation}
\begin{equation}
\begin{aligned}
\tr\left[\mathcal{F}\left(\left\{{\rm SWAP}_{i}| i=1,2,...n\right\}\right)
(\Sigma^{\otimes k}\otimes I_d^{\otimes n-k})\otimes(\Sigma^{\otimes k}\otimes I_d^{\otimes n-k})\right]=\tr\left[T^{k}\left(\left[\begin{matrix}
&d, & 0 \\
&0, & d^2
\end{matrix}\right]T\right)^{n-k}\right].
\end{aligned}
\end{equation}

We then calculate the term in Eq.~(\ref{Supp_Y_tk}),
\begin{equation}
\begin{aligned}
& \int_{\rm Haar} dY \tr\left[\mathcal{F}\left(\left\{{\rm SWAP}_{i}| i=1,2,...n\right\}\right)
(\Sigma^{\otimes k}\otimes Y)\otimes(\Sigma^{\otimes k}\otimes Y^\dag)\right]\\
=&\tr\left[T^{k}\left(\left[\begin{matrix}
&d, & 0 \\
&0, & 1
\end{matrix}\right]T\right)^{n-k}\right].
\end{aligned}
\end{equation}

Thus Eq.~(\ref{Supp_Y_tk}) turns out to be
 \begin{equation}\label{Supp_Y_tk_result}
\begin{aligned}
&\mathbb{E}_{M}\left[R_{M}(P_{\mathcal{S}})\right] \\
=& 1-\left(1-\frac{2}{d^k}\right)\left(1+\left(\frac{D^2d-d}{D^2d^2-1}\right)^n\right)-\tr\left[T^{k}\left(\left[\begin{matrix}
&d, & 0 \\
&0, & d^2
\end{matrix}\right]T\right)^{n-k}\right]-\tr\left[T^{k}\left(\left[\begin{matrix}
&d, & 0 \\
&0, & 1
\end{matrix}\right]T\right)^{n-k}\right],
\end{aligned}
\end{equation}
Utilizing the inequality $\tr[AB]\leq\tr[A]\tr[B]$ for positive semi-definite matrices $A$ and $B$, one then obtains for a general $k \in[1,n-1]$,
\begin{equation}\label{Supp_Y_tk_ineq}
\begin{aligned}
&\mathbb{E}_{M}\left[R_{M}(P_{\mathcal{S}})\right] \\
=& 1-\left(1-\frac{2}{d^k}\right)\left(1+\left(\frac{D^2d-d}{D^2d^2-1}\right)^n\right)-\tr\left[T^{k}\left(\left[\begin{matrix}
&d, & 0 \\
&0, & d^2
\end{matrix}\right]T\right)^{n-k}\right]-\tr\left[T^{k}\left(\left[\begin{matrix}
&d, & 0 \\
&0, & 1
\end{matrix}\right]T\right)^{n-k}\right]\\
\geq& 1-\left(1-\frac{2}{d^k}\right)\left(1+\left(\frac{D^2d-d}{D^2d^2-1}\right)^n\right)-\tr[T^k]\tr\left[\left(\left[\begin{matrix}
&d, & 0 \\
&0, & d^2
\end{matrix}\right]T\right)^{n-k}+\left(\left[\begin{matrix}
&d, & 0 \\
&0, & 1
\end{matrix}\right]T\right)^{n-k}\right]\\
=&1-\left(1-\frac{2}{d^k}\right)\left(1+\left(\frac{D^2d-d}{D^2d^2-1}\right)^n\right)- \left(\frac{(D+1)^k(Dd-1)^k+(D-1)^k(Dd+1)^k}{(D^2d^3-d)^k}\right)\\
&\cdot\left(1+\left(\frac{D^2d-d}{D^2d^2-1}\right)^{n-k}+\frac{1}{d^{n-k}}+\left(\frac{D^2-1}{D^2d^2-1}\right)^{n-k}\right).
\end{aligned}
\end{equation}

In the special case of $k=0$, the number of training samples $t_{k=0}=d^n-d^{n-k}=0$, one then has
\begin{equation}
\begin{aligned}
&\mathbb{E}_{M}\left[R_{M}(P_{\mathcal{S}})\right] \\
=& 1-\left(1-\frac{2}{d^k}\right)\left(1+\left(\frac{D^2d-d}{D^2d^2-1}\right)^n\right)-\tr\left[T^{k}\left(\left[\begin{matrix}
&d, & 0 \\
&0, & d^2
\end{matrix}\right]T\right)^{n-k}\right]-\tr\left[T^{k}\left(\left[\begin{matrix}
&d, & 0 \\
&0, & 1
\end{matrix}\right]T\right)^{n-k}\right]\\
=& 1-\left(1-\frac{2}{d^0}\right)\left(1+\left(\frac{D^2d-d}{D^2d^2-1}\right)^n\right)-\left(1+\left(\frac{D^2d-d}{D^2d^2-1}\right)^{n}+\frac{1}{d^{n}}+\left(\frac{D^2-1}{D^2d^2-1}\right)^{n}\right)\\
= &1+1+\left(\frac{D^2d-d}{D^2d^2-1}\right)^n-1-\left(\frac{D^2d-d}{D^2d^2-1}\right)^{n}-\frac{1}{d^{n}}-\left(\frac{D^2-1}{D^2d^2-1}\right)^{n}\\
=& 1-\frac{1}{d^{n}}-\left(\frac{D^2-1}{D^2d^2-1}\right)^{n}\\
\approx &1.
\end{aligned}
\end{equation}
The above result indicates that the average risk over $M$ approximates to $1$ for empty training set.

For $k=n$, the number of training samples $t_{k=n}=d^n-d^{n-n}\approx d^n$, one then has
\begin{equation}
\begin{aligned}
&\mathbb{E}_{M}\left[R_{M}(P_{\mathcal{S}})\right] \\
=& 1-\left(1-\frac{2}{d^k}\right)\left(1+\left(\frac{D^2d-d}{D^2d^2-1}\right)^n\right)-\tr\left[T^{k}\left(\left[\begin{matrix}
&d, & 0 \\
&0, & d^2
\end{matrix}\right]T\right)^{n-k}\right]-\tr\left[T^{k}\left(\left[\begin{matrix}
&d, & 0 \\
&0, & 1
\end{matrix}\right]T\right)^{n-k}\right]\\
=& 1-\left(1-\frac{2}{d^n}\right)\left(1+\left(\frac{D^2d-d}{D^2d^2-1}\right)^n\right)-\left(\frac{(D+1)^n(Dd-1)^n+(D-1)^n(Dd+1)^n}{(D^2d^3-d)^n}\right)\\
\approx&1-\left(1-\frac{2}{d^n}\right)\left(1+\frac{1}{d^n}\right)-\frac{2}{d^n}\\
\approx& 0.
\end{aligned}
\end{equation}
This indicates that the average risk over $M$ approximates to $0$ for full training set. 

The above results are independent of the training set $\mathcal{S}$, we thus obtain the following theorem.
\begin{theorem}[NFL theorem for MPS]\label{theorem:1D_Supp} 
Define the risk function $R_M(P_{\mathcal{S}})$ in Eq.~(\ref{Risk_error}) for learning a target $n$-qubit unitary $M$ based on the input of MPSs, where $P_{\mathcal{S}}$ represents the hypothesis unitary learned from the training set $\mathcal{S}$. Given a linear independent training set with size $t_k = d^n-d^{n-k}$, the integer $k\in [1,n-1]$, $d$ is the physical dimension of MPS, and $n$ denotes the qubit number of the system. The average risk is lower bounded by
\begin{equation}\label{NFL_Thm1_Supp}
\begin{aligned}
\mathbb{E}_{M,\mathcal{S}} \left[R_M(P_{\mathcal{S}})\right]\geq  1-(1-\frac{2}{d^k})(1+(dAB)^n)-(\frac{1}{d^n}+\frac{1}{d^k})(A^k+B^k)(1+(dAB)^{n-k}),
\end{aligned}
\end{equation} 
where  $A=\frac{D+1}{Dd+1}$, $B=\frac{D-1}{Dd-1}$, and $D$ is the bond dimension of the MPS. 
\end{theorem}

We numerically verify the result we obtain in Ineq.~(\ref{Supp_Y_tk_ineq}), see Fig.~\ref{Risk_M_Log_n40_D2_d2}.

\section{No-free-lunch theory for 2D tensor network machine learning model}
Now we proceed to prove the no-free-lunch theorem for the 2D TN based machine learning models. Specifically, we focus on the case of learning target quantum unitary operator based on the input of 2D projected entangled pair states (PEPSs), where the training states are encoded into unitary-embedded PEPs, and the variational structures are comprised of unitary quantum circuits. In constructing PEPSs, we apply a collection of $D^4d^2$-dimensional unitary tensors which are locally connected to neighboring tensors by virtual bond of dimension $D$ and the initial states $\ket{0}$ by physical bond of dimension $d$. This arrangement allows for the encoding of information into unitary tensors as input and output. The theorem that will be demonstrated within this chapter is restated as follows:
\begin{theorem}[NFL theorem for PEPS]\label{theorem:2D_Supp} 
Define the risk function $R_M(P_{\mathcal{S}})$ in Eq.~(\ref{Risk_error}) for learning a target $L^2$-qubit unitary $M$ based on the input of PEPSs, where $P_{\mathcal{S}}$ represents the hypothesis unitary learned from the training set $\mathcal{S}$. Given a linear independent training set with size $t_k = d^{L^2}-d^{L^2-k}$, the integer $k\in [1,L^2-1]$, $d$ is the physical dimension and virtual dimension of PEPS, respectively, and $L^2$ denotes the qubit number of the system. The average risk is lower bounded by
\begin{equation}\label{NFL_Thm2_Supp}
\begin{aligned}
&\mathbb{E}_{M,\mathcal{S}} \left[R_M(P_{\mathcal{S}})\right] \geq 1-(1+c(0.7)^L)\left[1-\frac{2}{d^{k}}+(1+\frac{1}{d^{L^2-k}})\left(\frac{2D^4 d-2}{D^4 d^3-d}\right)^{k}\left(\frac{1+D}{2D}\right)^{2k}(1+G(1/d,1/D^2))^{2l}\right],
\end{aligned}
\end{equation} 
where  $l=\lceil \sqrt{k} \rceil$, $G(q,p)=\frac{p}{2}\left(\sqrt{\frac{(1+q)(1+q-qp)}{1-q(2+p)+q^2(1-p)}}-1\right)$, $D$ is the bond dimension of the PEPS and $c$ is a constant. 
\end{theorem}

\subsection{The risk function} \label{supA}
We define the quantum Hilbert space for input and output samples by $\mathcal{H}_x$ and $\mathcal{H}_y$, and choose a subset $\mathcal{S}$ of the samples as training data, 
\begin{eqnarray}
	\mathcal{S}=\{(\ket{\psi_j},\ket{\phi_j}):\ket{\psi_j}\in\mathcal{H}_x,\ket{\phi_j}\in\mathcal{H}_y\}^t_{j=1}, 
\end{eqnarray}
with $\ket{\phi_j}=M\ket{\psi_j}/\bra{\psi_j}M^\dagger M\ket{\psi_j}$. Here, for a unitary embedded tensor-network state $\ket{\Psi}$ with bond dimension $D\ge 2$ and physical dimension $d\ge 2$ in a square lattice with volume $V=L\times L, L\ge 2$, the norm of the unitary embedded tensor-network state is exponentially concentrated around one as $L$ increases:
\begin{eqnarray}
	\Pr_M(|\Braket{\Psi|\Psi}|-1\ge \epsilon)\le O\left(\frac{\kappa^L}{\epsilon^2}\right).
\end{eqnarray}
To prove this Chebyshev inequality for PEPS, we just need to compute the variance of the norm, given by $Var(\Braket{\Psi|\Psi})=\overline{\Braket{\Psi|\Psi}^2}-\overline{\Braket{\Psi|\Psi}}^2$. Assuming that each local unitary forms a 1-design, it is easy to verify that $\overline{\Braket{\Psi|\Psi}}=\int dU_H \Braket{\Psi|\Psi}=1$. And since each local unitary of $\overline{\Braket{\Psi|\Psi}^2}$ forms a unitary 2-design, it will be shown that $\overline{\Braket{\Psi|\Psi}^2}\le 1+c(0.7^L)$ by Corollary \ref{h} in Sec.~\ref{Polyomino_SM}. So $Var(\Braket{\Psi|\Psi})=c(0.7^L)$ and the above inequality can be proved. 

Assuming that the model is perfectly trained, one then obtains the trained operator $P_{\mathcal{S}}$ as:
\begin{equation}
|\tilde{\phi_j}\rangle:=\frac{P_{\mathcal{S}}|\psi_j\rangle}{\langle\psi_j|P_{\mathcal{S}}^\dag P_{\mathcal{S}}|\psi_j\rangle} = e^{i\theta_j}\frac{M |\psi_j\rangle}{\langle\psi_j|M^\dag M|\psi_j\rangle}.
\end{equation}
As such, the risk function can be defined as
\begin{eqnarray}
	R_M(P_\mathcal{S})&=&\frac{1}{4}\int d\ket{\psi}\left\|\frac{M\ket{\psi}\bra{\psi}M^\dagger}{\bra{\psi}M^\dagger M\ket{\psi}^2}-\frac{P_\mathcal{S}\ket{\psi}\bra{\psi}P_\mathcal{S}^\dagger}{\bra{\psi}P_\mathcal{S}^\dagger P_\mathcal{S}\ket{\psi}^2}\right\|_1^2\nonumber\\
	&\rightarrow&1-\int d\ket{\psi}|\bra{\psi}M^\dagger P_\mathcal{S}\ket{\psi}|^2.
\end{eqnarray}

In the previous work \cite{Liu2023Theory}, calculating the second moment integral of 2D random PEPS is transformed into determining the partition function of 2D classical Ising model. The idea is that the second moment integral over Haar measure $\int dU_H U\otimes U^\dagger\otimes U\otimes U^\dagger$ forming the approximate unitary 2-design can be represented as a site on Ising model lattice with uniform external fields, similar to the 1D case. For a PEPS $\ket{\Psi}$ consisting of subsystems with periodic boundary conditions in a $L\times L$ lattice, the coordinate set of sites in the lattice is $T_L=\mathbb{Z}_L\times\mathbb{Z}_L$, and $\sigma_{x,y}$ is used to represent the spin state of the site at $(x,y)\in T$. Further, a configuration of the lattice is determined when each site takes a particular spin state, which can be written as a vector $\vec{\sigma}$. Through the above mapping and notation, we have
\begin{eqnarray}
	\int d\ket{\psi}|\bra{\psi}M^\dagger P_\mathcal{S}\ket{\psi}|^2=\sum_{\vec{\sigma}}\prod_{x,y={1...L}}^{}F(\sigma_{x,y},\sigma_{x+1,y},\sigma_{x,y+1})W\otimes W^\dagger,
\end{eqnarray}
where $W$ refers to $M^\dagger P_\mathcal{S}$. Here we define the transfer function as
\begin{equation*}
	F(\sigma_{x,y},\sigma_{x+1,y},\sigma_{x,y+1})W_{x,y}\otimes W_{x,y}^\dagger=\ipic{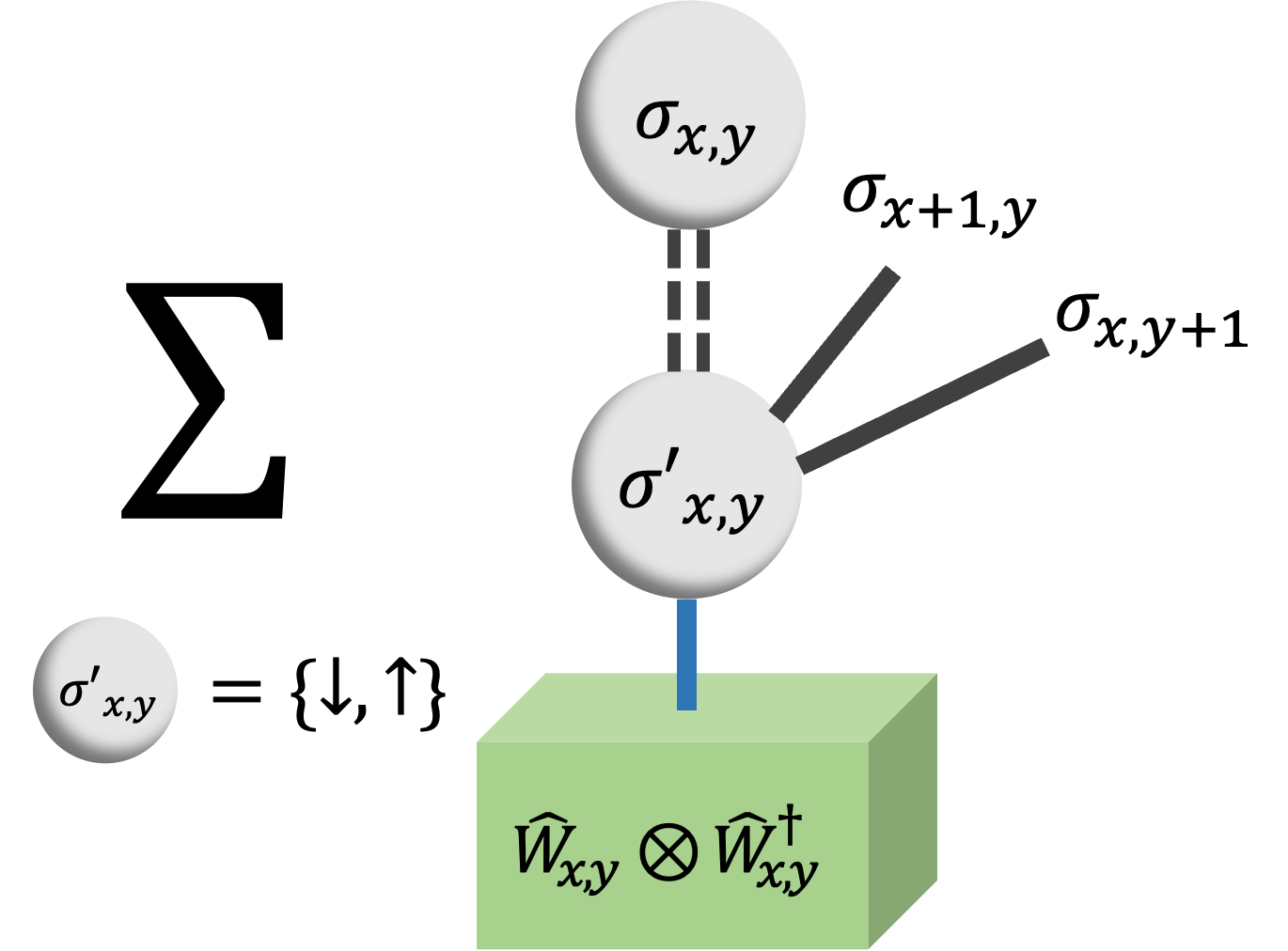}{0.25},
\end{equation*}
Yet we need to clarify that the notation $W_{x,y}$ just denotes how the local indices of $W_{i_{1,1}i_{1,2}...i_{x,y}...i_{L,L},j_{1,1}j_{1,2}...j_{x,y}...j_{L,L}}$ contract on the $(x,y)$ physical site, rather than the reduced matrix under partial trace, so ${\rm SWAP}_{x,y}$ denotes the swapping operation on the $(x,y)$ site of $W^\dag$ and $W$. 
The following values can be obtained:
\begin{eqnarray}
	F(\downarrow,\downarrow,\downarrow)W_{x,y}\otimes W_{x,y}^\dagger&=&\frac{D^4 d\tr(W_{x,y}\otimes W_{x,y}^\dagger)-\tr(\SWAP_{x,y}W_{x,y}\otimes W_{x,y}^\dagger)}{D^4 d^3-d},\nonumber\\
	F(\downarrow,\uparrow,\downarrow)W_{x,y}\otimes W_{x,y}^\dagger=F(\downarrow,\downarrow,\uparrow)W_{x,y}\otimes W_{x,y}^\dagger&=&\frac{D^3 d\tr(W_{x,y}\otimes W_{x,y}^\dagger)-D\tr(\SWAP_{x,y}W_{x,y}\otimes W_{x,y}^\dagger)}{D^4 d^3-d},\nonumber\\
	F(\downarrow,\uparrow,\uparrow)W_{x,y}\otimes W_{x,y}^\dagger&=&\frac{D^2 d\tr(W_{x,y}\otimes W_{x,y}^\dagger)-D^2 \tr(\SWAP_{x,y}W_{x,y}\otimes W_{x,y}^\dagger)}{D^4 d^3-d},\nonumber\\
	F(\uparrow,\uparrow,\uparrow)W_{x,y}\otimes W_{x,y}^\dagger&=&\frac{D^4 d\tr(\SWAP_{x,y}W_{x,y}\otimes W_{x,y}^\dagger)-\tr(W_{x,y}\otimes W_{x,y}^\dagger)}{D^4 d^3-d},\nonumber\\
	F(\uparrow,\uparrow,\downarrow)W_{x,y}\otimes W_{x,y}^\dagger=F(\uparrow,\downarrow,\uparrow)W_{x,y}\otimes W_{x,y}^\dagger&=&\frac{D^3 d\tr(\SWAP_{x,y}W_{x,y}\otimes W_{x,y}^\dagger)-D\tr(W_{x,y}\otimes W_{x,y}^\dagger)}{D^4 d^3-d},\nonumber\\
	F(\uparrow,\downarrow,\downarrow)W_{x,y}\otimes W_{x,y}^\dagger&=&\frac{D^2 d\tr(\SWAP_{x,y}W_{x,y}\otimes W_{x,y}^\dagger)-D^2 \tr(W_{x,y}\otimes W_{x,y}^\dagger)}{D^4 d^3-d}.
\end{eqnarray}
To prove Theorem \ref{theorem:2D}, one  needs to calculate  the average of the risk function: 
\begin{eqnarray}
    \mathbb{E}_{M,\mathcal{S}} \left[R_M(P_{\mathcal{S}})\right]&=&1-\int dW \sum_{\vec{\sigma}}\prod_{x,y}F(\sigma_{x,y},\sigma_{x+1,y},\sigma_{x,y+1})W\otimes W^\dagger \nonumber\\
    &=&1-\int dW \sum_{\vec{\sigma}}\mathcal{A}[\vec{\sigma}]W\otimes W^\dagger\nonumber\\
    &=&1-Z. 
\end{eqnarray}
$\mathcal{A}[\vec{\sigma}]W\otimes W^\dagger$ denotes the contribution of one of the various configurations $\vec{\sigma}$, which now is a function with $D,d$ and $\SWAP_{x,y}$ acting on $W\otimes W^\dagger$. Thus we can obtain the partition function by summing of all the contributions of various $\vec{\sigma}$. In the subsequent sections, we will demonstrate how to perform the calculation using the polyomino method.

\subsection{The polyominoes}\label{Polyomino_SM}

To address the 2D TN problem, we initially introduce a statistical model for enumerating lattice configurations on a two-dimensional plane, known as the polyomino. The focus of this model is on directed figures that encompass a specific number of sites within an infinite square grid. We will find that the partition function of the two-dimensional Ising model can be translated into a problem sets on a periodic plane of an infinite square lattice grid, making the application of polyomino calculation method particularly effective. We first provide some definitions related to polyominoes \cite{Liu2023Theory}, and then map our problem to this framework. The formal definition of the directed polyomino(shown in Fig.~\ref{poly}) is listed as follows: 
\begin{definition}
	(Directed polyomino). A directed polyomino is a finite subset $\vec{\tau}\subseteq\mathbb{Z}\times\mathbb{Z}$ ($\mathbb{Z}$ denotes the set of integers) in the Euclidean plane rooted at $(x_0,y_0)$ such that
	\begin{itemize}
		\item $(x_0,y_0)\in \vec{\tau}$,
		\item for $(x, y) \in \vec{\tau}$ and $(x, y) \neq\left(x_{0}, y_{0}\right)$, at least one of $(x+1, y)\in\vec{\tau}$ and $(x, y+1)\in\vec{\tau}$ are in $\vec{\tau}$ as well. 
	\end{itemize}
\end{definition}

\begin{figure}[htbp]
    \centering
    \includegraphics[width=0.9\linewidth]{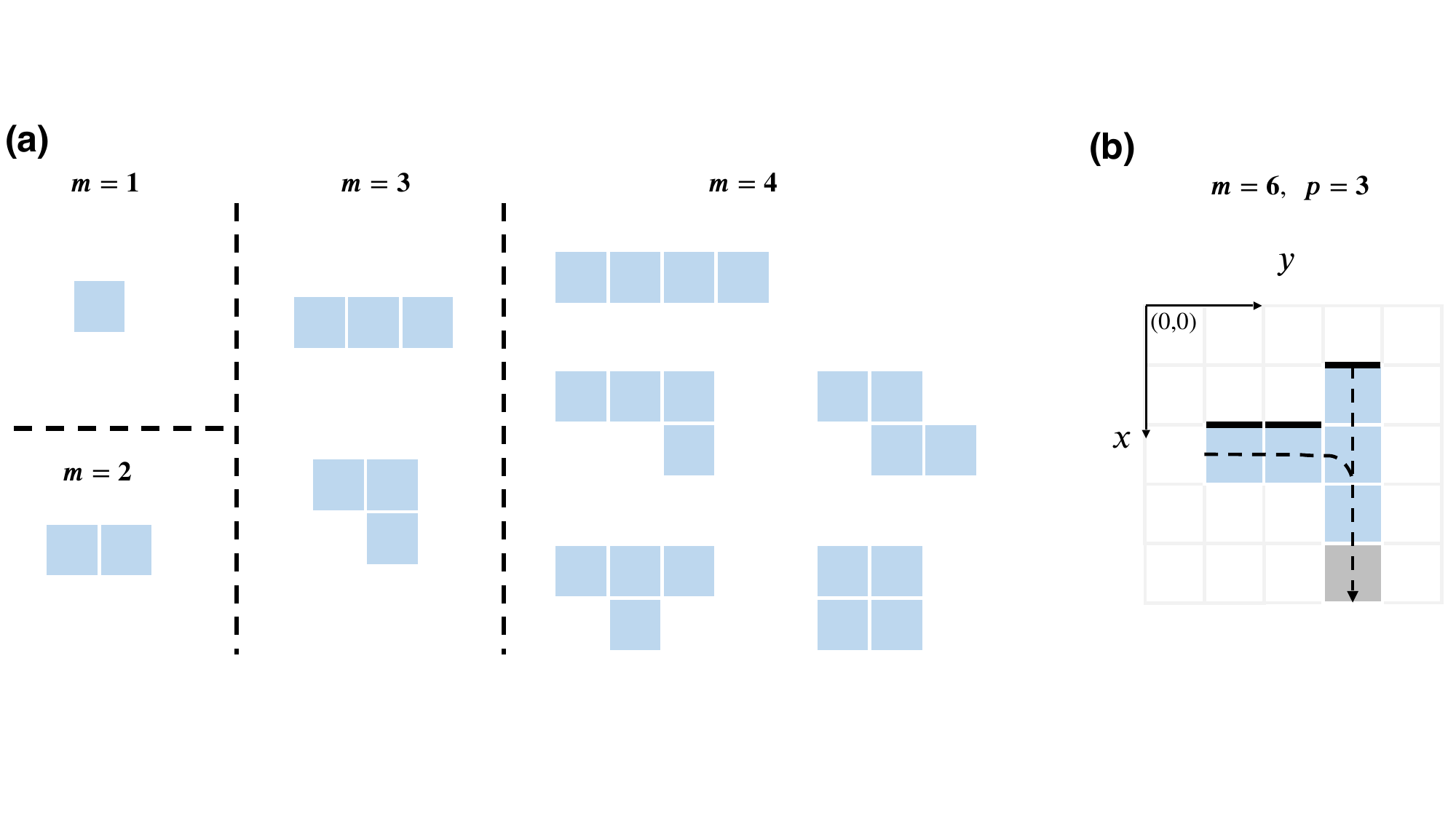}
    \vspace{-40pt}
    \caption{Illustrations of the general polyomino and the directed polyomino. (a) Illustration of the various types of general polyominoes with the area ranging from 1 to 4. Here we regard the different rotations of a polyomino as being of the same type. (b) Illustration of a directed polyomino rooted at the gray site with area m = 6, perimeter p = 14, and upper perimeter n = 3 (depicted as the black lines).}
    \label{poly}
\end{figure} 

The polyominoes can be characterized by their areas, perimeters and upper perimeters as their properties, whose formal definitions are as follows:
\begin{definition}
	(area, perimeter and upper perimeter). Let $\vec{\tau}$ be a directed polyomino.
	\begin{itemize}
		\item The area $m$ of $\vec{\tau}$ is the size of the polyomino, i.e.$|\vec{\tau}|$.
		\item The perimeter $p$ of $\vec{\tau}$ is the number of edges on the boundary of $\vec{\tau}$. Formally speaking, $p=|\{(x,y)\in \mathbb{Z}\times \mathbb{Z}:\tau_{x,y}\neq \tau_{x+1,y}\}|+|\{(x,y)\in \mathbb{Z}\times \mathbb{Z}:\tau_{x,y}\neq \tau_{x,y+1}\}|$.
		\item The upper perimeter $n$ of $\vec{\tau}$ is the number of horizontal edges on the top boundary of $\vec{\tau}$. Formally speaking,$n=|\{(x,y)\in \mathbb{Z}\times \mathbb{Z}:\tau_{x,y}=0,\tau_{x+1,y}=1\}|$. 
	\end{itemize}
\end{definition}
Then we can enumerate directed polyominoes exactly via their generating function by the following lemma:
\begin{lemma}
	Let $D_{m,n}$ be the number of directed polyominoes rooted at (0,0) with area and upper perimeter $m,n$ respectively. $p$ and $q$ are two variables. We introduce the generating function of $D_{m,n}$ as
	\begin{equation}
		G(q,p)=\sum_{m,n}D_{m,n}q^m p^n.
	\end{equation}
	When $p,q\in (0,1]$ such that $|q(2+p)-q^2(1-p)|<1$, the power series converges to a finite value
	\begin{equation}
		G(q,p)=\frac{p}{2}\left(\sqrt{\frac{(1+q)(1+q-qp)}{1-q(2+p)+q^2(1-p)}}-1\right).
	\end{equation}
\end{lemma}
This lemma implies that when $q,p\in (0,1]$ such that $|q(2+p)+q^2(1-p)|\le 1$, the power series $\sum_{m,n}D_{m,n}q^mp^n$ converges to a finite value $G(q,p)$. We will see the problem of partition function for two-dimensional lattice having a similar power-series form reduces to determining the number of configurations $D_{m,n}$. 

\subsection{Apply polyominoes to the calculation of PEPS 2-moment}
In Sec. \ref{supA}, we have mentioned that calculating the integral $\int dU_H |\braket{\Psi|\hat{O}|\Psi}|^2$ of PEPS is transformed into calculating the partition function
\begin{eqnarray}\label{LL}
	\int dU_H |\braket{\Psi|\hat{O}|\Psi}|^2 =\sum_{\vec{\sigma}}\prod_{x,y}F(\sigma_{x,y},\sigma_{x+1,y},\sigma_{x,y+1})\hat{O}\otimes \hat{O}^\dagger.
\end{eqnarray}
Here $\sum_{\vec{\sigma}}$ means the sum over all configuration $\vec{\sigma}$, which is a simple form of $\sum_{\substack{\sigma_{x,y}=\{\uparrow,\downarrow\} \\ (x,y) \in T_L}}$. If $\hat{O}=I$, then we denote $\prod_{x,y}F(\sigma_{x,y},\sigma_{x+1,y},\sigma_{x,y+1})\hat{O}\otimes \hat{O}^\dagger=\prod_{x,y}f(\sigma_{x,y},\sigma_{x+1,y},\sigma_{x,y+1})$, where
\begin{equation}\label{F}
	F(\sigma_{x,y},\sigma_{x+1,y},\sigma_{x,y+1})\hat{O}\otimes \hat{O}^\dagger=\ipic{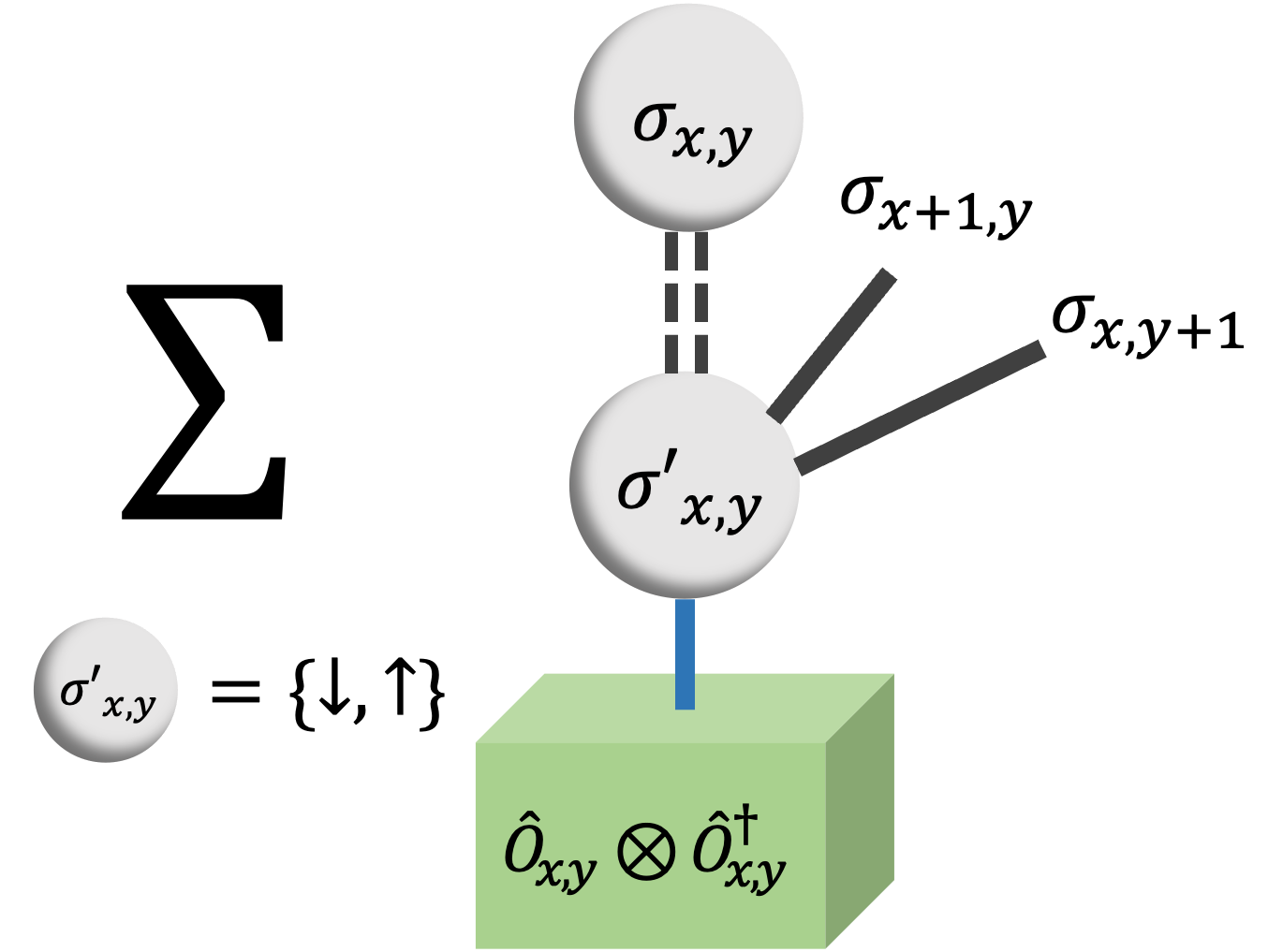}{0.25}, f(\sigma_{x,y},\sigma_{x+1,y},\sigma_{x,y+1})=\ipic{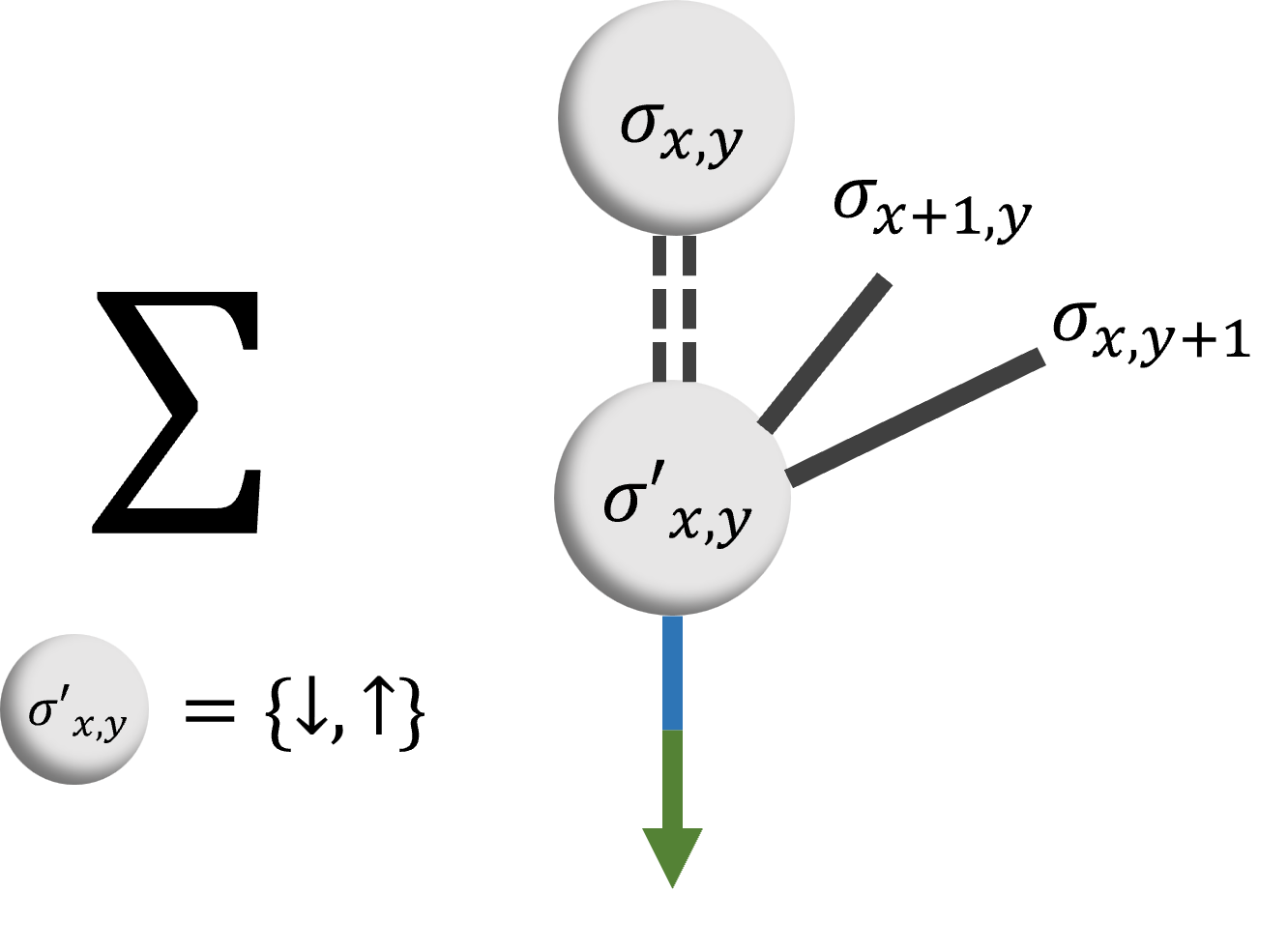}{0.25}.
\end{equation}
Then
\begin{eqnarray}
	\begin{aligned}
		f(\downarrow,\downarrow,\downarrow)&=1,\\
		f(\downarrow,\uparrow,\downarrow)=f(\downarrow,\downarrow,\uparrow)&=\frac{D^3 d^3-Dd}{D^4 d^3-d},\\
		f(\downarrow,\uparrow,\uparrow)&=\frac{D^2 d^3-D^2 d}{D^4 d^3-d},\\
		f(\uparrow,\uparrow,\uparrow)&=\frac{D^4 d^2-d^2}{D^4 d^3-d},\\
		f(\uparrow,\uparrow,\downarrow)=f(\uparrow,\downarrow,\uparrow)&=\frac{D^3 d^2-Dd^2}{D^4 d^3-d},\\
		f(\uparrow,\downarrow,\downarrow)&=0.
	\end{aligned}
\end{eqnarray}

As mentioned above, the partition function can be viewed as the sum of the contribution of different $\vec{\sigma}$. Owing to $f(\uparrow,\downarrow,\downarrow)=0$, one conly needs to consider the configurations that every $\uparrow$ on lattice must have either a right or lower $\uparrow$-neighbor; otherwise, the contribution of the configuration $\prod_{x,y}f(\sigma_{x,y},\sigma_{x+1,y},\sigma_{x,y+1})$ will be 0. The Ising model is defined on a lattice with periodic boundaries, so these $\uparrow$ sites must form at least one loop on a torus.  We refer to these as excited-spin-strings(ESSs), with a formal definition provided in Definition \ref{defofESS}. And for a lattice with $L^2$ sites, each cycle ESS forms a loop with at least $L$ areas. 

For a site $(x,y)$ with $\sigma_{x,y}=\uparrow$: 
\begin{itemize}
	\item If $\sigma_{x+1,y}=\uparrow$, we mark the deriction as $(x,y)\to (x+1,y)$.
	\item If $\sigma_{x+1,y}=\downarrow$ but $\sigma_{x,y+1}=\uparrow$, we mark the deriction as $(x,y)\to (x,y+1)$.
\end{itemize}
In this way, we have constructed a directed graph $G=(V,E)$, where the vertex set $V$ is the set of all $\uparrow$-sites.
\begin{definition}
	(root). A root is a vertex $(x_0,y_0)$ in a directed graph that every vertex $(x,y)$ can be reached through the direction on the graph.
\end{definition}

Here we explain in details the concept of roots. The roots of polyominoes defined on the plain are different from those on the torus. Directed graphs consisting of $\uparrow$ are shown in Fig.~\ref{ESS}. There is only one root if the last vertex has no further extension, while under periodic boundary conditions, the vertices on the boundary can be linked to their counterparts on the opposite side. This connectivity introduces the possibility of the $\uparrow$ forming a closed loop within the directed graph, leading to that each $\uparrow$ can always reach other $\uparrow$ in the loop through the direction of the graph. In such a scenario, all $\uparrow$ in the loop are the roots of the directed graph. 

\begin{figure}[htbp]
	\centering
	\includegraphics[width=0.8\linewidth]{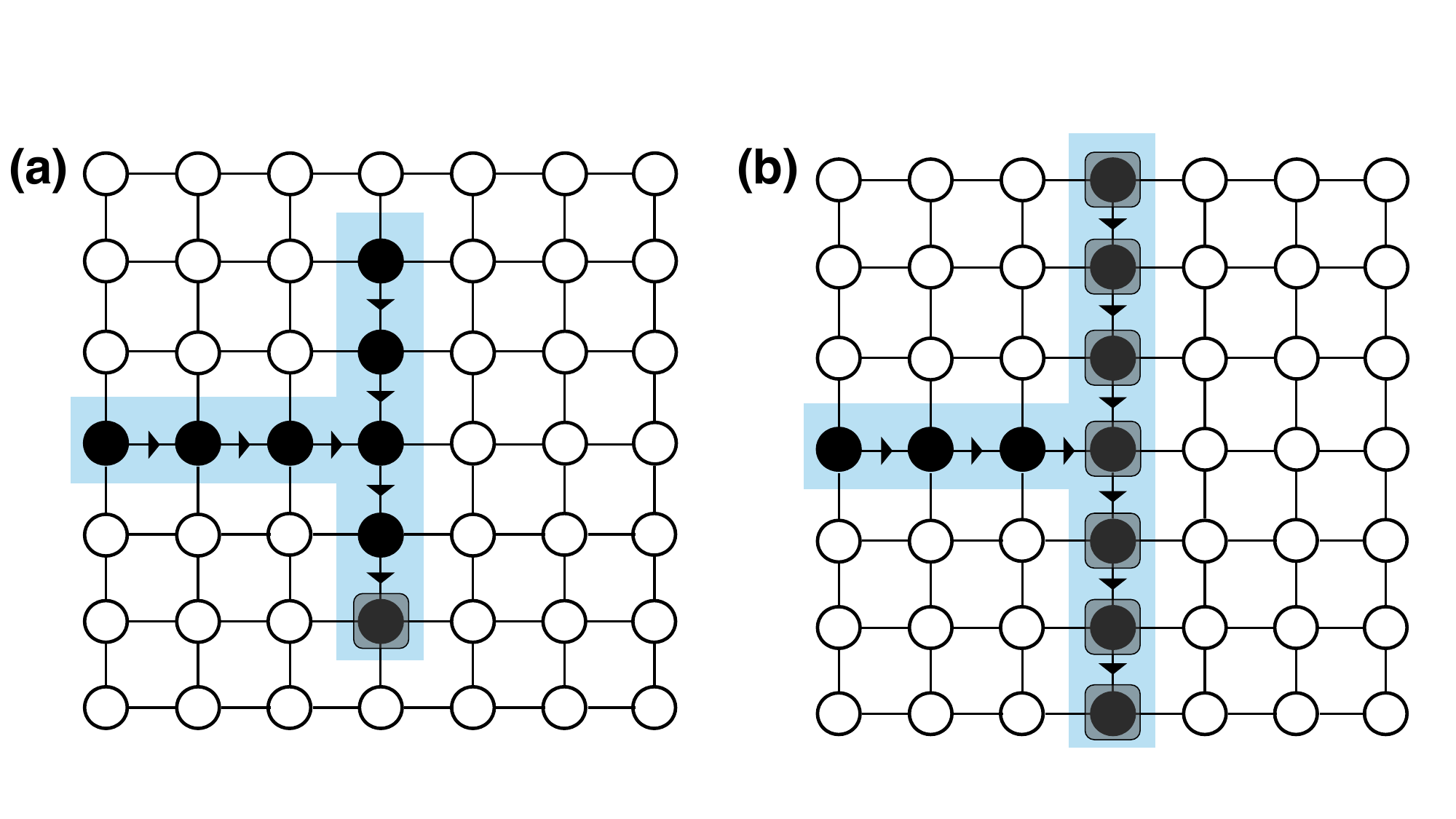}
	\caption{Directed graphs consist of $\uparrow$. Solid and hollow points separately represent the sites of $\uparrow$ and $\downarrow$, while the sites in grey square represents the root of the graph. (a) An example of an ESS which have only one root. (b) An example of a cycle ESS which have more than one root, where all sites on cycle are its roots.}
	\label{ESS}
\end{figure}

Now we have the below definition of the excited-spin-string(ESS):
\begin{definition}\label{defofESS}
	(ESS \& cycle ESS) An ESS in the toric plane $T_L$ rooted at $(x_0,y_0)$ is a connected subset of $\vec{\tau_T}\subseteq \mathbb{Z}_L\times\mathbb{Z}_L$, such that
	\begin{itemize}
		\item $(x_0,y_0)\in \vec{\tau_T}$, 
		\item for $(x,y)\in \vec{\tau_T},(x,y)\neq(x_0,y_0)$, at least one of $(x+1, y)\in\vec{\tau_T}$ and $(x, y+1)\in\vec{\tau_T}$. 
	\end{itemize}
	A cycle excited-spin-string(cycle ESS) is an ESS in which the number of its roots are more than one. 
\end{definition}
It is appropriate to treat an ESS as a polyomino defined on a torus. We take the count of ESSs originating at (0,0) with area $m$ and upper perimeter $n$ by $\tilde{D}_{m,n}$ for the number of torus polyominoes. The following theorem establishes the connection between $\tilde{D}_{m,n}$ for torus polyominoes and $D_{m,n}$ for plane polyominoes: 
\begin{theorem}\label{bridge}
	(Lemma 5 and Lemma 6 in Ref. \cite{Liu2023Theory}) Let $\vec{\sigma}\in \Sigma$ be a configuration with area $m$, perimeter $p$ and upper perimeter $n$, then
	\begin{eqnarray}
		\prod_{x,y}f(\sigma_{x,y},\sigma_{x+1,y},\sigma_{x,y+1})\le q_a^m q_p^p\le q_a^m q_p^{4n},
	\end{eqnarray} 
	where $q_a=f(\uparrow,\uparrow,\uparrow)\le \frac{1}{d}, q_p=f(\downarrow,\downarrow,\uparrow)\le \frac{1}{D}$. And by applying the toric polyominoes with $L^2$ sites we have
	\begin{eqnarray}
		\sum_{\vec{\sigma}}\prod_{x,y}f(\sigma_{x,y},\sigma_{x+1,y},\sigma_{x,y+1})\le \sum_{m,n}\tilde{D}_{m,n}q_a^m q_p^p\le 2\sum_{m,n}\tilde{D}_{m,n}q_a^m q_p^{4n}\le \sum_{m,n}\tilde{D}_{m,n}q_a^m q_p^{2n},
	\end{eqnarray}
	where 
	\begin{eqnarray}
		\tilde{D}_{m,n}\le \sum_{k=1}^{[m/L]}\sum_{c=0}^{k}\sum_{\substack{m_1,...,m_k\ge L\\ m_1+...+m_k=m\\ n_1+...n_k=n+c}}\prod_{i=1}^{k}(L^2D_{m_i,n_i}),
	\end{eqnarray}
	$k$ is the number of ESSs in each configuration.
\end{theorem}
This theorem demonstrates that every contribution of each configuration $\vec{\sigma}$ can be limited by their area $m$ and upper perimeter $n$ with two parameter $q_a, q_p$. Consequently, the partition function is effectively bounded by considering all the potential configurations. Moreover, the number $\tilde{D}_{m,n}$ of torus polyominoes can be linked with $D_{m,n}$ via the above inequality allowing us to determine it. According to that we can figure out the upper bound of the 2-moment integral $\int dU_H |\braket{\Psi|\Psi}|^2$ by Corollary \ref{h}. 
\begin{corollary}\label{h}
	(Theorem 1 in Ref. \cite{Liu2023Theory})The upper bound of the partition function 
	$$\int dU_H |\Braket{\Psi|\Psi}|^2=\sum_{\vec{\sigma}}\prod_{x,y}f(\sigma_{x,y},\sigma_{x+1,y},\sigma_{x,y+1})$$
	is converge to $1+c(0.7)^L$ with a constant c.
\end{corollary}

Here, we only state the conclusions we need. One can find the details of proof for Theorem \ref{bridge} and Corollary \ref{h} in Ref. \cite{Liu2023Theory}. By applying this corollary, one can prove the concentration of PEPS based on the Chebyshev inequality.

\subsection{Average risk of 2D tensor-network ML model}
In the previous section, we have connected the partition function with the torus polyomino. Now we calculate  the average of the risk function: 
\begin{eqnarray}
    \mathbb{E}_{M,\mathcal{S}} \left[R_M(P_{\mathcal{S}})\right]&=&1-\int dW \sum_{\vec{\sigma}}\prod_{x,y}F(\sigma_{x,y},\sigma_{x+1,y},\sigma_{x,y+1})W\otimes W^\dagger\nonumber\\
    &=&1-\int dW \sum_{\vec{\sigma}}\mathcal{A}[\vec{\sigma}]W\otimes W^\dagger
\end{eqnarray}

Here for simplicity, we consider the size $t$ of training set be  $t=d^{L^2}-d^{L^2-k}$. Following the similar way as in Eq.~\ref{Supp_Y_tk}, we consider the case that a subspace comprising k sites on lattice can be fully trained. By labeling the sites from a particular starting site and arranging them counterclockwise around the center, the trained space can be amalgamated into the $k$ trained sites that approximates a $l\times l$ square. This is because both its upper and left boundaries are no more than $l$, where $l$ can be expressed as $\lceil \sqrt{k} \rceil$. We can refer to these $k$ sites as the trained zone $A$, as shown in Fig. \ref{root}. Here we denote the coordinates set of the $k$ trained sites in A  as $T_A$. 
\begin{figure}[htbp]
	\centering
	\includegraphics[width=1\linewidth]{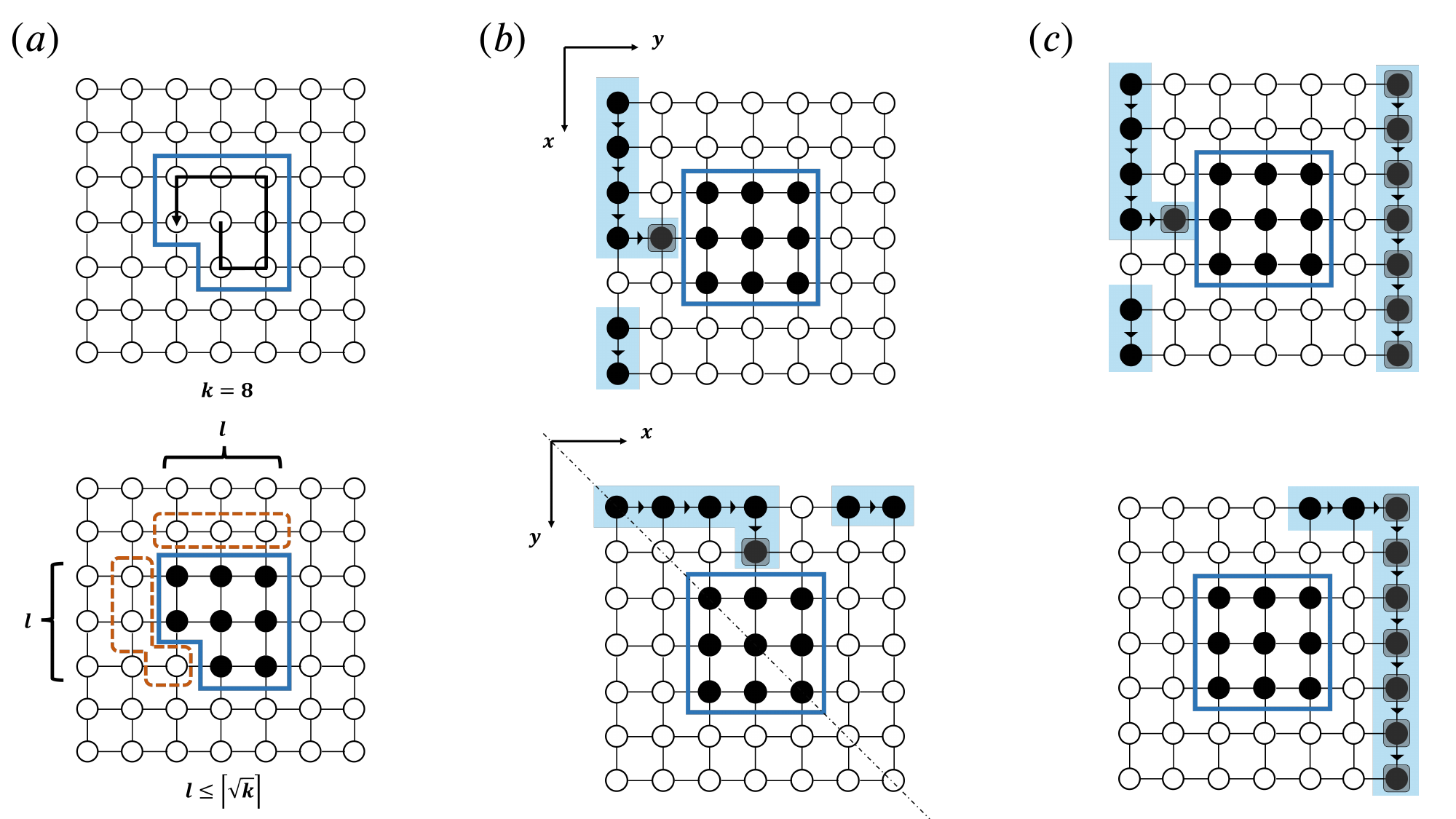}
 	\caption{Illustrations of the zone $A$ of trained sites and configurations outside A. (a) The $t$ training samples are arranged at $k$ sites labeled by $i=1,2,...,k$ in counterclockwise order, which forms a region whose upper and left boundary both no more than $l$. The sites within the dashed box indicate the candidate sites for the roots of ESSs out of A. (b) Illustration of the contribution of the i-th ESS rooted at the left boundary of A and the upper boundary of A, the two configurations are equivalent under diagonal reflection. (c) Examples of two categories of configurations distinguished by whether they contain cycle ESSs. These are different from the configurations in (b) which only contain an ESS rooted at the sites adjacent to the boundary of A. }
	\label{root}
\end{figure} 

Suppose that the learned unitary operator $M$ satisfies the condition $M^\dagger P_\mathcal{S}=W=e^{i\theta I_t}\oplus Y$, we employ the same technique as in the 1D case to analyze and manipulate this operator, 
\begin{eqnarray}\label{divide}
	W\otimes W^\dagger &=& (I_d^{\otimes n})^{\otimes 2}-I_d^{\otimes n}\otimes \Sigma^{\otimes k}\otimes I_d^{\otimes n-k}-\Sigma^{\otimes k}\otimes I_d^{\otimes n}\otimes I_d^{\otimes n-k}\nonumber\\
    & &+(\Sigma^{\otimes k}\otimes I_d^{\otimes n-k})^{\otimes 2}+\Sigma^{\otimes k}\otimes Y\otimes \Sigma^{\otimes k}\otimes Y^\dagger.
\end{eqnarray}
For arbitrary local tensor, we have
\begin{eqnarray}
	&&\tr\left[{\rm SWAP}_i I_{d,i}\otimes I_{d,i}\right] = d,\qquad \tr\left[ I_{d,i}\otimes I_{d,i}\right] = d^2;\\
	&&\tr\left[{\rm SWAP}_i I_{d,i}\otimes \Sigma_i\right] = 1,\qquad \tr\left[ I_{d,i}\otimes \Sigma_i\right] = d;\\
	&&\tr\left[{\rm SWAP}_i \Sigma_i\otimes I_{d,i}\right] = 1,\qquad \tr\left[ \Sigma_i\otimes I_{d,i}\right] = d;\\
	&&\tr\left[{\rm SWAP}_i \Sigma_i\otimes \Sigma_i\right] = 1,\qquad \tr\left[ \Sigma_i\otimes \Sigma_i\right] = 1.
\end{eqnarray}
 By substituting the above formulas into Eq.~(\ref{F}) , one obtains 
\begin{equation}
	\ipic{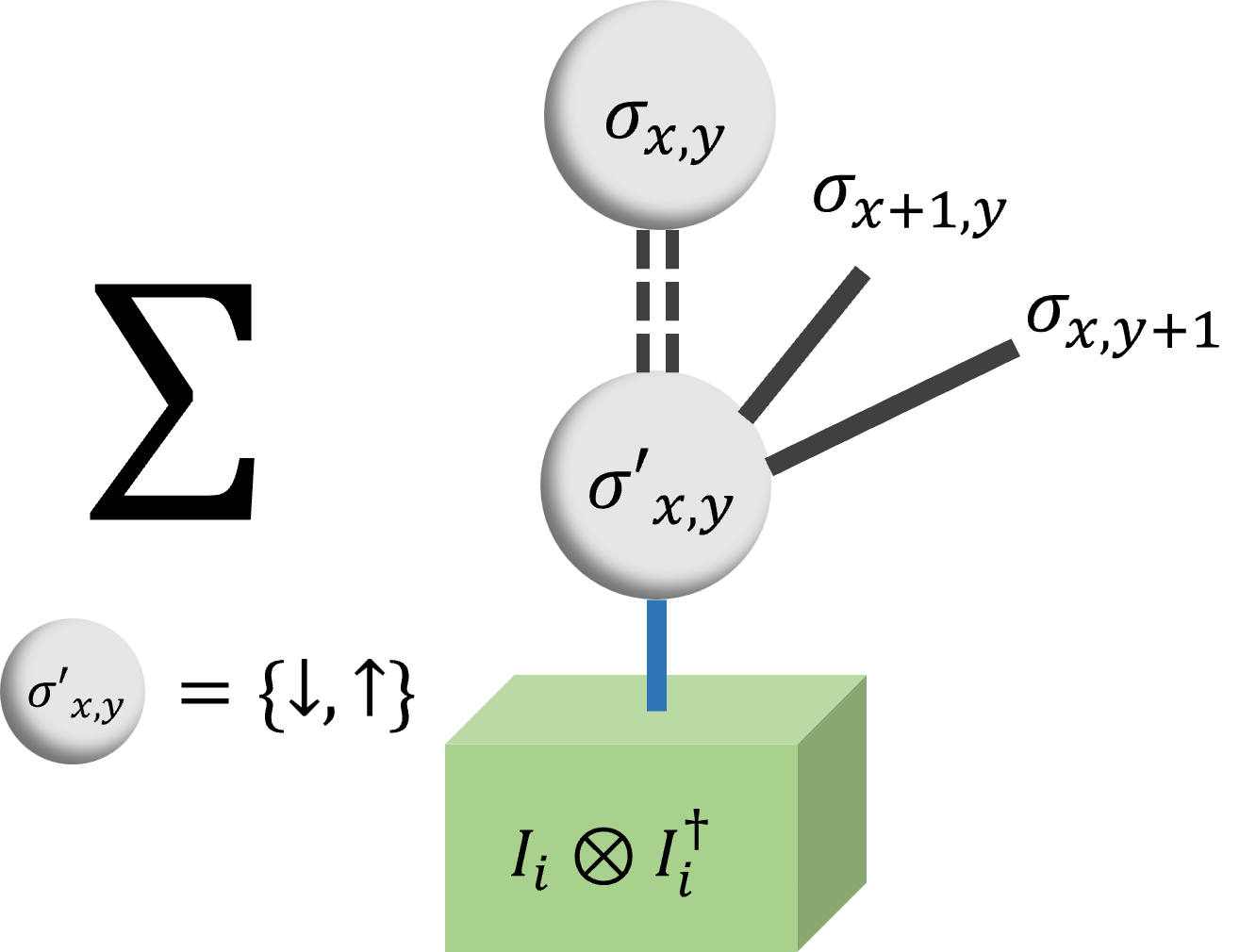}{0.25}=f(\cdot),\qquad\ipic{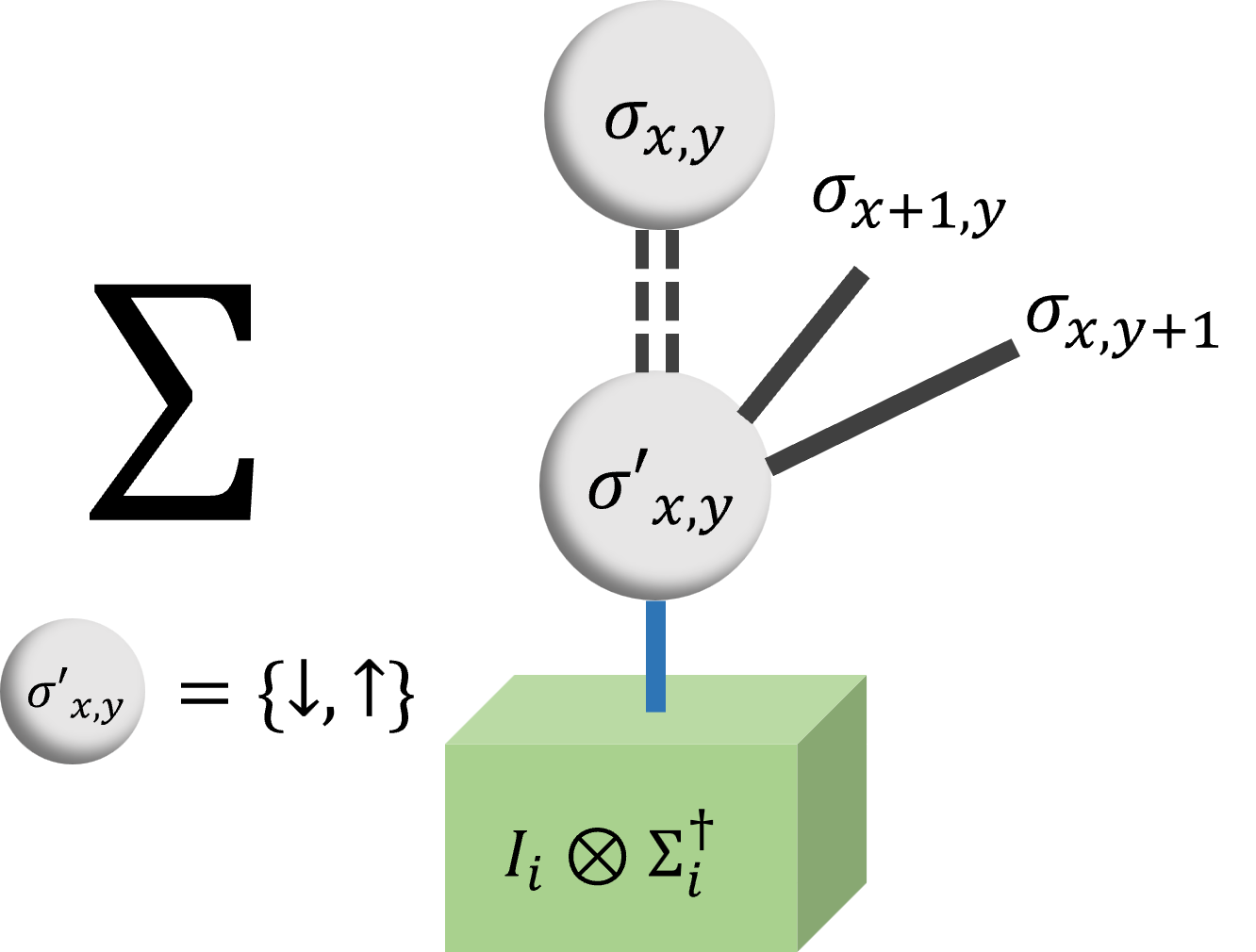}{0.25}=\frac{1}{d}f(\cdot),\qquad\ipic{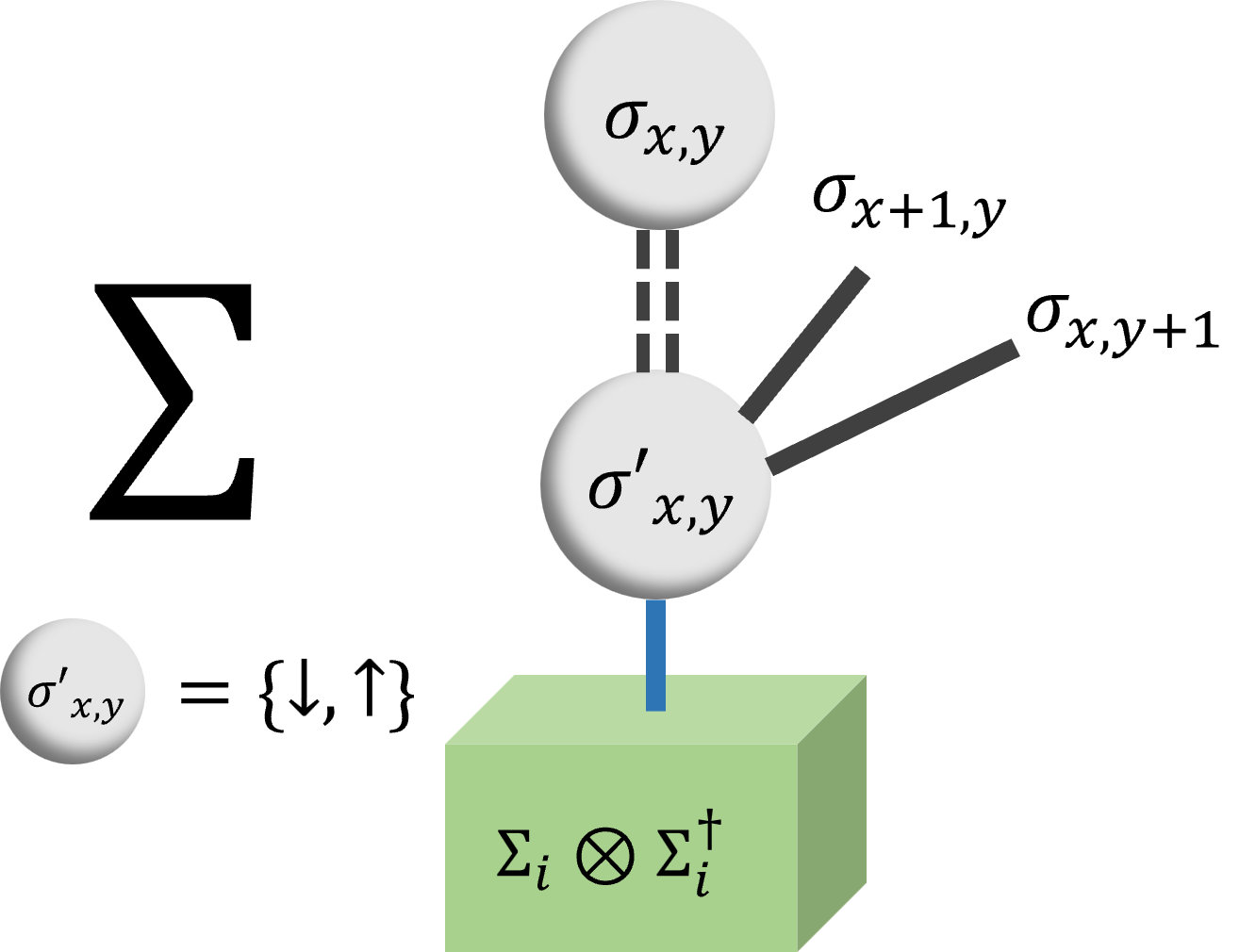}{0.25}=g(\cdot).
\end{equation}

Then base on Eq.~(\ref{divide}), we can divide $Z=\int dW \sum_{\sigma_{x,y}={\uparrow,\downarrow}}\mathcal{A}[\vec{\sigma}]W\otimes W^\dagger$  into the following five terms. For the first three items, by applying Corollary \ref{h}, we  obtain
\begin{eqnarray}
	Z_1=\sum_{\vec{\sigma}}\mathcal{A}[\vec{\sigma}](I_d^{\otimes n})^{\otimes 2}&=&\sum_{\vec{\sigma}}\prod_{x,y={1...L}}^{}f(\sigma_{x,y},\sigma_{x+1,y},\sigma_{x,y+1})\le 1+c(0.7^L),\\
	Z_2=\sum_{\vec{\sigma}}\mathcal{A}[\vec{\sigma}]I_d^{\otimes n}\otimes \Sigma^{\otimes k}\otimes I_d^{\otimes n-k}&=&\sum_{\vec{\sigma}}\frac{1}{d^k}\prod_{x,y={1...L}}^{}f(\sigma_{x,y},\sigma_{x+1,y},\sigma_{x,y+1})\le \frac{1}{d^k}(1+c(0.7^L)),\\
	Z_3=\sum_{\vec{\sigma}}\mathcal{A}[\vec{\sigma}]\Sigma^{\otimes k}\otimes I_d^{\otimes n}\otimes I_d^{\otimes n-k} &=&\sum_{\vec{\sigma}}\frac{1}{d^k}\prod_{x,y={1...L}}^{}f(\sigma_{x,y},\sigma_{x+1,y},\sigma_{x,y+1})\le \frac{1}{d^k}(1+c(0.7^L)).
\end{eqnarray}
For the latter two terms, we have
\begin{eqnarray}\label{A}
	Z_4&=&\sum_{\vec{\sigma}}\mathcal{A}[\vec{\sigma}](\Sigma_i^{\otimes k})^{\otimes2}\otimes(I_d^{\otimes L^2-k})^{\otimes 2}\nonumber\\
	&=&\sum_{\vec{\sigma}}\prod_{(x,y)\in T_A}g(\sigma_{x,y},\sigma_{x+1,y},\sigma_{x,y+1})\prod_{(x,y)\in T\setminus T_A}f(\sigma_{x,y},\sigma_{x+1,y},\sigma_{x,y+1})
\end{eqnarray}
and
\begin{eqnarray}\label{B}
	Z_5&=&\int dY \sum_{\vec{\sigma}}\mathcal{A}[\vec{\sigma}](\Sigma_i^{\otimes k})^{\otimes2}\otimes Y\otimes Y^\dagger\nonumber\\
	&=&\sum_{\vec{\sigma}}\prod_{(x,y)\in T_A}g(\sigma_{x,y},\sigma_{x+1,y},\sigma_{x,y+1})\frac{1}{d^{L^2-k}}\prod_{(x,y)\in T\setminus T_A}f(-\sigma_{x,y},-\sigma_{x+1,y},-\sigma_{x,y+1}),
\end{eqnarray}
where we denote $-\sigma_{x,y}$ as the opposite spin state of $\sigma_{x,y}$, and
\begin{eqnarray}
    g(\uparrow,\uparrow,\uparrow)=g(\downarrow,\downarrow,\downarrow)&=&\frac{D^4 d-1}{D^4 d^3-d};\nonumber\\
	g(\uparrow,\uparrow,\downarrow)=g(\uparrow,\downarrow,\uparrow)=g(\downarrow,\uparrow,\downarrow)=g(\downarrow,\downarrow,\uparrow)&=&\frac{D^3 d-D}{D^4 d^3-d} \le \frac{1}{D}g(\uparrow,\uparrow,\uparrow);\nonumber\\
	g(\uparrow,\downarrow,\downarrow)=g(\downarrow,\uparrow,\uparrow)&=&\frac{D^2 d-D^2}{D^4 d^3-d} \le \frac{1}{D^2}g(\uparrow,\uparrow,\uparrow).
\end{eqnarray}

Utilizing the inequality
\begin{eqnarray}
    &&\sum_{\vec{\sigma}}\prod_{(x,y)\in T_A}g(\sigma_{x,y},\sigma_{x+1,y},\sigma_{x,y+1})\prod_{(x,y)\in T\setminus T_A}f(\sigma_{x,y},\sigma_{x+1,y},\sigma_{x,y+1})\nonumber\\
    &&\le \left(\sum_{\vec{\sigma}}\prod_{(x,y)\in T_A}g(\sigma_{x,y},\sigma_{x+1,y},\sigma_{x,y+1})\right)\left(\sum_{\vec{\sigma}}\prod_{(x,y)\in T\setminus T_A}f(\sigma_{x,y},\sigma_{x+1,y},\sigma_{x,y+1})\right)
\end{eqnarray}
one can represent the terms $Z_4$ and $Z_5$ by the combination of values derived from $g(\sigma_{x,y},\sigma_{x+1,y},\sigma_{x,y+1})$ for $(x,y)\in T_A$, which correspond to the sites within zone $A$, and $f(\sigma_{x,y},\sigma_{x+1,y},\sigma_{x,y+1})$ for $(x,y)\in T\setminus T_A$, which correspond to the sites outside the zone $A$. We will first introduce how to calculate the values of the sites within zone $A$.

For the convenience of representing $\prod_{(x,y)}g(\sigma_{x,y},\sigma_{x+1,y},\sigma_{x,y+1})$, let $G(\sigma_1, ..., \sigma_n)$ denote an arbitrary product of $g(\cdot)$ with the elements $\sigma_1, ..., \sigma_n$. Since $g(\sigma_{x,y},\sigma_{x+1,y},\sigma_{x,y+1})=g(-\sigma_{x,y},-\sigma_{x+1,y},-\sigma_{x,y+1})$ for arbitrary $g(\cdot)$, for fixed element $\sigma_k$, we find
\begin{eqnarray}
    \sum_{\sigma_1}...\sum_{\sigma_n}G(\sigma_1,..., \sigma_k=\uparrow, ..., \sigma_n)&=&\sum_{\sigma_1}...\sum_{\sigma_n}G(-\sigma_1,..., \sigma_k=\downarrow, ..., -\sigma_n)\nonumber\\
    &=&\sum_{-\sigma_1}...\sum_{-\sigma_n}G(-\sigma_1,..., \sigma_k=\downarrow, ..., -\sigma_n)\nonumber\\
    &=&\sum_{\sigma_1}...\sum_{\sigma_n}G(\sigma_1,..., \sigma_k=\downarrow, ..., \sigma_n), 
\end{eqnarray}
and
\begin{equation}
    \sum_{\sigma_1}...\sum_{\sigma_n}G(\sigma_1,..., \uparrow, ..., \sigma_n)+\sum_{\sigma_1}...\sum_{\sigma_n}G(\sigma_1,..., \downarrow, ..., \sigma_n)=\sum_{\sigma_1}...\sum_{\sigma_n}\sum_{\sigma_k}G(\sigma_1,..., \sigma_k, ..., \sigma_n),
\end{equation}
where $\sum_{\sigma_i}$ means $\sum_{\sigma_i=\{\uparrow,\downarrow\}}$. Thus we have
\begin{equation}\label{sum_G}
\begin{aligned}
    \sum_{\sigma_1}...\sum_{\sigma_n}G(\sigma_1,..., \uparrow, ..., \sigma_n)&=\sum_{\sigma_1}...\sum_{\sigma_n}G(\sigma_1,..., \downarrow, ..., \sigma_n)\\
    &=\frac{1}{2}\sum_{\sigma_1}...\sum_{\sigma_n}\sum_{\sigma_k}G(\sigma_1,..., \sigma_k, ..., \sigma_n). 
\end{aligned}
\end{equation}
Additionally, the following equation holds, 
\begin{equation}
\begin{aligned}
    \sum_{\sigma_i}\sum_{\sigma_j}\sum_{\sigma_k}g(\sigma_i,\sigma_j, \sigma_k) &\le \left(2+\frac{4}{D}+\frac{2}{D^2}\right)g(\uparrow,\uparrow,\uparrow) \\
    &= 2\left(\frac{1+D}{D}\right)^2\left(\frac{D^4 d-1}{D^4 d^3-d}\right). 
\end{aligned}
\end{equation}

To determine the values of the sites within zone $A$, we start by designating the top-left corner site as $(x_0, y_0)$ and define $T_1$ as the set containing the sites $(x_0, y_0)$, $(x_0+1, y_0)$, and $(x_0, y_0+1)$. These three sites correspond to the elements of the factor $g(\sigma_{x_0, y_0}, \sigma_{x_0+1, y_0}, \sigma_{x_0, y_0+1})$. We initiate the process by expanding the summation over the spins $\sigma_{x_0, y_0}$, $\sigma_{x_0+1, y_0}$, and $\sigma_{x_0, y_0+1}$. At this initial stage, only one factor, $g(\sigma_{x_0, y_0}, \sigma_{x_0+1, y_0}, \sigma_{x_0, y_0+1})$, depends on $\sigma_{x_0, y_0}$. We denote $\prod_{(x,y)\in T_A\setminus(x_0,y_0)}g(\sigma_{x,y},\sigma_{x+1,y},\sigma_{x,y+1})$ as $G(\vec{\sigma}')$. Since $\sigma_{x_0+1, y_0}$ and $\sigma_{x_0, y_0+1}$ must not be included in the same $g(\cdot)$ factor except for $g(\sigma_{x_0, y_0}, \sigma_{x_0+1, y_0}, \sigma_{x_0, y_0+1})$, we can separate $G(\vec{\sigma}')$ into two distinct factors, $G(\vec{\sigma_1}')$ and $G(\vec{\sigma_2}')$, which contain $\sigma_{x_0+1, y_0}$ and $\sigma_{x_0, y_0+1}$, respectively. With this setup, we can proceed with the following steps,

\begin{eqnarray}
    & &\sum_{\vec{\sigma}}\prod_{(x,y)\in T_A}g(\sigma_{x,y},\sigma_{x+1,y},\sigma_{x,y+1})\nonumber\\
    &=&\sum_{\substack{\sigma_{x,y}=\{\uparrow,\downarrow\} \\ (x,y)\in T\setminus T_1 }}\Bigg[\sum_{\sigma_{x_0,y_0}}\sum_{\sigma_{x_0+1,y_0}}\sum_{\sigma_{x_0,y_0+1}}g(\sigma_{x_0,y_0},\sigma_{x_0+1,y_0},\sigma_{x_0,y_0+1})G(\vec{\sigma}')\Bigg]\nonumber\\
    &=& \sum_{\substack{\sigma_{x,y}=\{\uparrow,\downarrow\} \\ (x,y)\in T\setminus T_1}}\Bigg[g(\uparrow,\uparrow,\uparrow)G(\vec{\sigma_1}'|_{\sigma_{x_0+1,y_0}=\uparrow})G(\vec{\sigma_2}'|_{\sigma_{x_0,y_0+1}=\uparrow}) \nonumber\\
    & &\qquad\qquad+g(\uparrow,\uparrow,\downarrow)G(\vec{\sigma_1}'|_{\sigma_{x_0+1,y_0}=\uparrow})G(\vec{\sigma_2}'|_{\sigma_{x_0,y_0+1}=\downarrow})\nonumber\\
    & &\qquad\qquad...\nonumber\\
    & &\qquad\qquad+g(\downarrow,\downarrow,\downarrow)G(\vec{\sigma_1}'|_{\sigma_{x_0+1,y_0}=\downarrow})G(\vec{\sigma_2}'|_{\sigma_{x_0,y_0+1}=\downarrow})\Bigg] \nonumber\\
    &=&\sum_{\substack{\sigma_{x,y}=\{\uparrow,\downarrow\} \\ (x,y)\in T\setminus (x_0,y_0)}}\frac{1}{4}\Bigg[\sum_{\sigma_{x_0,y_0}}\sum_{\sigma_{x_0+1,y_0}}\sum_{\sigma_{x_0,y_0+1}}g(\sigma_{x_0,y_0},\sigma_{x_0+1,y_0},\sigma_{x_0,y_0+1})\Bigg]G(\vec{\sigma_1}')G(\vec{\sigma_2}')\nonumber\\
    &=&\sum_{\substack{\sigma_{x,y}=\{\uparrow,\downarrow\} \\ (x,y)\in T\setminus (x_0,y_0)}}\frac{1}{4}\Bigg[\sum_{\sigma_{x_0,y_0}}\sum_{\sigma_{x_0+1,y_0}}\sum_{\sigma_{x_0,y_0+1}}g(\sigma_{x_0,y_0},\sigma_{x_0+1,y_0},\sigma_{x_0,y_0+1})\Bigg]\prod_{(x,y)\in T_A\setminus (x_0,y_0)}g(\sigma_{x,y},\sigma_{x+1,y},\sigma_{x,y+1}) \nonumber\\
    &\le &\sum_{\substack{\sigma_{x,y}=\{\uparrow,\downarrow\} \\ (x,y)\in T\setminus (x_0,y_0)}}\Bigg[\frac{1}{2}\left(\frac{1+D}{D}\right)^2\left(\frac{D^4 d-1}{D^4 d^3-d}\right)\Bigg]\prod_{(x,y)\in T_A\setminus (x_0,y_0)}g(\sigma_{x,y},\sigma_{x+1,y},\sigma_{x,y+1}) \nonumber\\
    &\le& \left(\frac{1+D}{D}\right)^{k^2}\left(\frac{D^4 d-1}{2D^4 d^3-2d}\right)^k.\label{sum_g_peps}
\end{eqnarray}
For the second equal sign in Eq.~(\ref{sum_g_peps}), the eight terms are derived by expanding the three summations within the square brackets. For the third equal sign in Eq.~(\ref{sum_g_peps}), each occurrence of $G(\vec{\sigma_1}'|_{\sigma_{x_0+1,y_0}=\uparrow})$ or $G(\vec{\sigma_1}'|_{\sigma_{x_0+1,y_0}=\downarrow})$ is replaced with $\frac{1}{2}\sum_{\sigma_{x_0+1,y_0}}G(\vec{\sigma_1}')$, and similarly, $G(\vec{\sigma_2}'|_{\sigma_{x_0,y_0+1}=\uparrow})$ or $G(\vec{\sigma_2}'|_{\sigma_{x_0,y_0+1}=\downarrow})$ is substituted with $\frac{1}{2}\sum_{\sigma_{x_0,y_0+1}}G(\vec{\sigma_2}')$, in accordance with Eq.~(\ref{sum_G}). This restoration of the two summations $\sum_{\sigma_{x_0+1,y_0}}\sum_{\sigma_{x_0,y_0+1}}$ then allows $G(\vec{\sigma_1}')G(\vec{\sigma_2}')$ to be combined into a single factor, $G(\vec{\sigma}')$. By applying this process iteratively to the subsequent terms of $g(\sigma_{x,y},\sigma_{x+1,y},\sigma_{x,y+1})$, we obtain  the expression $\left(\frac{1+D}{D}\right)^{k^2}\left(\frac{D^4 d-1}{2D^4 d^3-2d}\right)^k$. Then we have

\begin{eqnarray}
    Z_4 &=& \sum_{\vec{\sigma}}\prod_{(x,y)\in T_A}g(\sigma_{x,y},\sigma_{x+1,y},\sigma_{x,y+1})\prod_{(x,y)\in T\setminus T_A}f(\sigma_{x,y},\sigma_{x+1,y},\sigma_{x,y+1}) \nonumber\\
    &\le& \sum_{\vec{\sigma}}\prod_{(x,y)\in T_A}g(\sigma_{x,y},\sigma_{x+1,y},\sigma_{x,y+1})\sum_{\vec{\sigma}}\prod_{(x,y)\in T\setminus T_A}f(\sigma_{x,y},\sigma_{x+1,y},\sigma_{x,y+1}) \nonumber\\
    &\le& \left(\frac{1+D}{D}\right)^{k^2}\left(\frac{D^4 d-1}{2D^4 d^3-2d}\right)^k\sum_{\vec{\sigma}}\prod_{(x,y)\in T\setminus T_A}f(\sigma_{x,y},\sigma_{x+1,y},\sigma_{x,y+1}).\label{Z4_SM}
\end{eqnarray}

Now we consider the value of the remaining $L^2-k$ sites outside $A$ for $Z_4$, i.e., the term $\sum_{\vec{\sigma}}\prod_{(x,y)\in T\setminus T_A}f(\sigma_{x,y},\sigma_{x+1,y},\sigma_{x,y+1})$. 
Given the condition that $f(\uparrow,\downarrow,\downarrow)=0$, here we only need to identify those non-zero configurations outside $A$. If a site $(x,y)$ on upper boundary of zone $A$ is $\uparrow$, there could exist non-zero configurations containing an ESS rooted at $(x-1,y)$. Similar case occurs for a $\uparrow$ on left boundary. Thus, different from the $L\times L$ case in Eq.~
(\ref{LL}), the non-zero configurations in the $L^2-k$ lattice outside $A$ in Eq.~(\ref{Z4_SM}) include not only the cycle ESSs but also the ESSs possibly rooted at the sites adjacent to the upper and left boundary of $A$, shown in Fig. \ref{root}(a).

These configurations can be divided into two parts,  depending on whether they contain cycle ESSs or not, as shown in Fig. \ref{root}(c). For the configuration without cycle ESS, all the ESSs of toric polyominoes are rooted at the sites adjacent to the upper and left boundary of A. The contribution of the i-th ESS rooted at the sites adjacent to the left boundary of A is less than $\sum_{m_i,n_i}D_{m_i,n_i}q_a^{m_i}q_p^{2n_i}$. We assume that all potential sites can serve as roots, so as to cover all possible configurations. By considering all cases of $k$, roots positions and the overlaps with boundary $0\le c\le k$, we obtain
\begin{eqnarray}
    &&\sum_{k}^{l} \begin{pmatrix}
    l\\k
    \end{pmatrix}\sum_{m_i,n_i}\prod_{k}^{i}D_{m_i,n_i}q_a^{m_i}q_p^{2n_i} \nonumber\\
    &=&\sum_{k}^{l} \begin{pmatrix}
    l\\k
    \end{pmatrix} G(q_a,q_p^2)^k\nonumber\\
    &=&\left(1+G(q_a,q_p^2)\right)^{l}.
\end{eqnarray}

By reflecting the $x$ and $y$ axes across the diagonal line and altering the priority of 'right' and 'down' directions of the directed graphs on torus, one can calculate the case of ESSs rooted at the upper boundary, as shown in Fig. \ref{root}(b). Consequently 
\begin{equation}
    Z_{4 (out)}^I=\left(1+G(q_a,q_p^2)\right)^{2l}
\end{equation}

Now we calculate  the contribution of  configurations containing at least one cycle ESS. Let the toric polyominoes with area $m$ and upper perimeter $n$ be decomposed into $K$ ESSs, each with area $m_i$ and upper perimeter $n_i$ $(i=1,2,...,K)$, respectively. We assume that the first $K_1$ ESSs are cycle ESSs. The remaining $K_2$ ESSs are rooted at the neighbors of upper and left boundary of $A$. We then obtain

\textbullet$m_1,...,m_{K_1}\ge L$; 

\textbullet$m_1+...+m_{K_1}+...+m_{K_1+K_2}=m$; 

\textbullet$n_1+...+n_{K_1+K_2}=n+c$; \\
Thus the number of toric polyominoes outside $A$ is
\begin{equation}
    E_{m,n}\le \sum_{K_1=1}^{m/L}\sum_{K_2=0}^{2l}\begin{pmatrix}2l\\K_2\end{pmatrix}\sum_{c=0}^{K_1+K_2}\sum_{\substack{m_1,...,m_{K_1}\ge L\\m_1+...+m_{K_1+K_2}=m\\ n_1+...+n_{K_1+K_2}=n+c}}\prod_{i=1}^{K_1}((L^2-k)D_{m_i,n_i})\prod_{i=K_1+1}^{K_1+K_2}D_{m_i,n_i}.
\end{equation}
Then the generating function obeys
\begin{eqnarray}
    Z_{4 (out)}^{II}&=&2\sum_{m,n}E_{m,n}q_a^m q_p^{4n}\nonumber\\
    &\le &2\sum_{K_1=1}^{m/L}\sum_{K_2=0}^{2l}\begin{pmatrix}2l\\K_2\end{pmatrix}\sum_{c=0}^{K_1+K_2}q_p^{-4c}\sum_{\substack{m_1,...,m_{K_1}\ge L\\m_1+...+m_{K_1+K_2}=m\\ n_1+...+n_{K_1+K_2}=n+c}}\prod_{i=1}^{K_1}((L^2-k)D_{m_i,n_i}q_a^{m_i} q_p^{4n_i})\prod_{i=K_1+1}^{K_1+K_2}D_{m_i,n_i}q_a^{m_i} q_p^{4n_i}\nonumber\\
    &\le &2\sum_{K_1=1}^{m/L}\sum_{K_2=0}^{2l}\begin{pmatrix}2l\\K_2\end{pmatrix}\frac{q_p^{-4K_1-4K_2}}{1-q_p^4}\sum_{\substack{m_{K_1+1},...,m_{K_1+K_2}\\ n_{K_1+1},...,n_{K_1+K_2}}}\left((L^2-k)\sum_{m\ge L,n}D_{m,n}q_a^{m} q_p^{4n}\right)^{K_1}\prod_{i=K_1+1}^{K_1+K_2}D_{m_i,n_i}q_a^{m_i} q_p^{2n_i}\nonumber\\
	&\le &\sum_{K_1=1}^{L}\frac{2}{1-q_p^4}\left((L^2-k)q_p^{-4}\sum_{m\ge L,n}D_{m,n}q_a^{m} q_p^{4n}\right)^{K_1}\sum_{K_2=0}^{2l}\begin{pmatrix}2l\\K_2\end{pmatrix}\sum_{\substack{m_{K_1+1},...,m_{K_1+K_2}\\ n_{K_1+1},...,n_{K_1+K_2}}}\prod_{i=K_1+1}^{K_1+K_2}q_p^{-4K_2}G(q_a,q_p^2)^{2l}\nonumber\\
	&\le &\frac{2L}{1-q_p^4}\max_{K_1\le L}\left((L^2-k) q_p^{-4}\sum_{m\ge L,n}D_{m,n}q_a^{m} q_p^{4n}\right)^{K_1}(1+q_p^{-4}G(q_a,q_p^2))^{2l}\nonumber\\
	&\le &c(0.7)^L(1+q_p^{-4}G(q_a,q_p^2))^{2l}.\label{Generate_Function_Z4_SM}
\end{eqnarray}
By definition, $q_a\le 1/2, q_p^4\le 1/16$. We can then obtain the last line of inequality in  Eq.~(\ref{Generate_Function_Z4_SM}) through the following inequality
    \begin{equation}
    \begin{aligned}
    q_{p}^{-4} \sum_{m \geq L, n} D_{m, n} q_{a}^{m} q_{p}^{4 n}
    \leq & \sum_{m \geq L, n} D_{m, n}(0.5)^{m}(1 / 16)^{n-1} \\
    \leq & 16(0.5 / 0.72)^{L} \sum_{m \geq L, n} D_{m, n}(0.72)^{m}(1 / 16)^{n} \\
    \leq & 16(0.695)^{L} G(0.72,1 / 16) \\
    \leq & 27(0.695)^{L} \leq c(0.7)^L.
    \end{aligned}
    \end{equation}
Actually, for enough large $L$, $27 (L^{2}-k)(0.695)^{L}<1$, we have $\frac{54L(L^{2}-k)}{1-q_{p}^{4}}(0.695)^{L}<c(0.7)^{L}$ for some constant $c$. Otherwise, for a small $L$, one can always find a constant $c$ such that the factor is less than $c(0.7)^{L}$ as well. Thus the term $Z_4$ is bounded by
\begin{eqnarray}
	Z_4&\le&\left(\frac{D^4 d-1}{2D^4 d^3-2d}\right)^{k}\left(\frac{1+D}{D}\right)^{2k} (Z_{4(out)}^{I}+Z_{4(out)}^{II})\nonumber\\
	&\le&\left(\frac{2D^4 d-2}{D^4 d^3-d}\right)^{k}\left(\frac{1+D}{2D}\right)^{2k}\left[\left(1+G(q_a,q_p^2)\right)^{2l}+(c(0.7)^L)(1+q_p^{-4}G(q_a,q_p^2))^{2l}\right]\nonumber\\
	&\le&\left(\frac{2D^4 d-2}{D^4 d^3-d}\right)^{k}\left(\frac{1+D}{2D}\right)^{2k}(1+c(0.7)^L)(1+G(1/d,1/D^2))^{2l}.
\end{eqnarray}

For the function $f(-\sigma_{x,y},-\sigma_{x+1,y},-\sigma_{x,y+1})$ in the term $Z_5$, we can just take the sets of $\downarrow$ as the ESSs and obtain the result in a similar approach. We obtain that
\begin{eqnarray}
	Z_5 &= \frac{1}{d^{L^2-k}} Z_4.
\end{eqnarray}
\subsection{Results and Discussion}
Summing over all of the contributing terms $\{Z_1,Z_2,Z_3,Z_4,Z_5\}$, one can obtain the average risk
\begin{equation}
\begin{aligned}
	\mathbb{E}_{M,\mathcal{S}} \left[R_M(P_{\mathcal{S}})\right] =& 1-(Z_1+Z_2+Z_3+Z_4+Z_5)\\
    = & 1-(1+c(0.7)^L)\left[1-\frac{2}{d^{k}}+(1+\frac{1}{d^{L^2-k}})\left(\frac{2D^4 d-2}{D^4 d^3-d}\right)^{k}\left(\frac{1+D}{2D}\right)^{2k}(1+G(1/d,1/D^2))^{2l}\right].
\end{aligned}
\end{equation}
Thus we have obtained Theorem~\ref{theorem:2D}. In the thermodynamic limit where $L\to \infty$, one can ignore the higher-order terms and obtains: 
\begin{align}
    \mathbb{E}_{M,\mathcal{S}} \left[R_M(P_{\mathcal{S}})\right] \ge \frac{2}{d^{k}}-(1+\frac{1}{d^{L^2-k}})\left(\frac{2D^4 d-2}{D^4 d^3-d}\right)^{k}\left(\frac{1+D}{2D}\right)^{2k}(1+G(1/d,1/D^2))^{2l}.
\end{align}
This result is shown in Fig.~\ref{Risk_M_Log_n40_D2_d2}.

We consider two special cases, one case is the empty training set and the other one is the full training set. 
For empty training set with $k = 0$, the lower bound of the average  risk  is 
\begin{eqnarray}
	\mathbb{E}_{M,\mathcal{S}} \left[R_M(P_{\mathcal{S}})\right] \ge 1-\left(1-2+(1+0)\cdot 1\right)=1.
\end{eqnarray}
For full training set with $k =L^2$, the lower bound of the average  risk  is 
\begin{eqnarray}
	\mathbb{E}_{M,\mathcal{S}} \left[R_M(P_{\mathcal{S}})\right] \ge 1-\left(1-0+(1+1)\cdot 0\right)=0.
\end{eqnarray}

\begin{figure}
	\centering
	\includegraphics[width=0.6\linewidth]{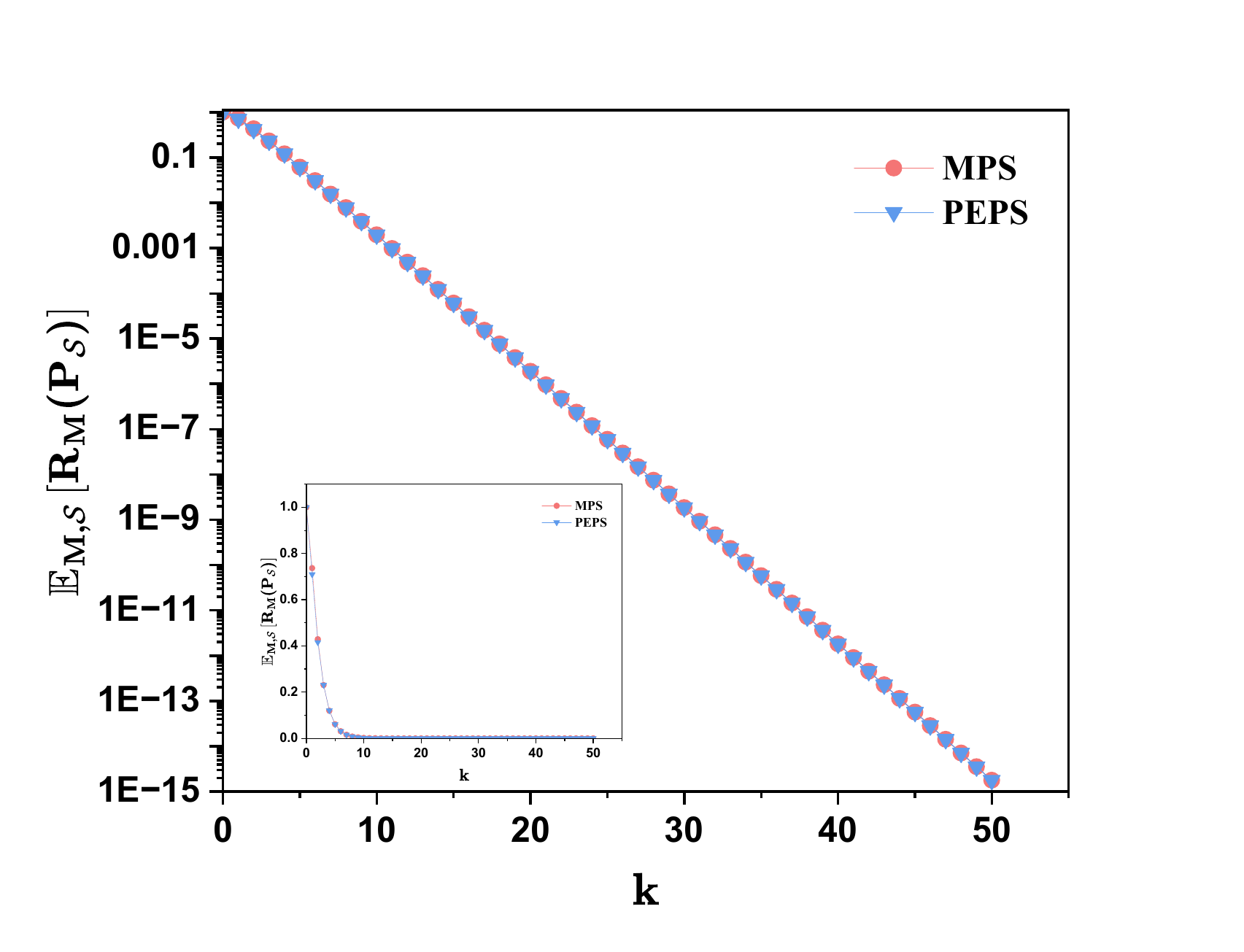}
	\caption{Analytical results of the lower bounds of the average risk versus the size of training set $t_k = d^n-d^{n-k}$, both for the MPS (Theorem 1) and PEPS (Theorem 2). The parameters we choose: qudit number $n=50$, bond dimension $D=2$, physical dimension $d=2$, the integer $k\in[1,49]$. We see that with $k$ increasing from 1 to 49, i.e., the number of training samples $t_k = d^n-d^{n-k}$ increases,  the average risk over $M$ decays  towards zero with respect to $k$.}
\label{Risk_M_Log_n40_D2_d2}
\end{figure} 
Moreover, taking $n = 4$ as an instance, numerical simulations for MPS machine learning model was carried out regarding the average risks of 10 unitary matrices, as shown in Fig.~\ref{convergence}. Through varying the training error within the practical learning process, diverse average risk curves were acquired. It is evident that as the training error diminishes, the average risk curves decline correspondingly and approach the analytical results. Our findings fully accord with the implications of being the average lower bound of the risk function.
\begin{figure}
	\centering
	\includegraphics[width=0.6\linewidth]{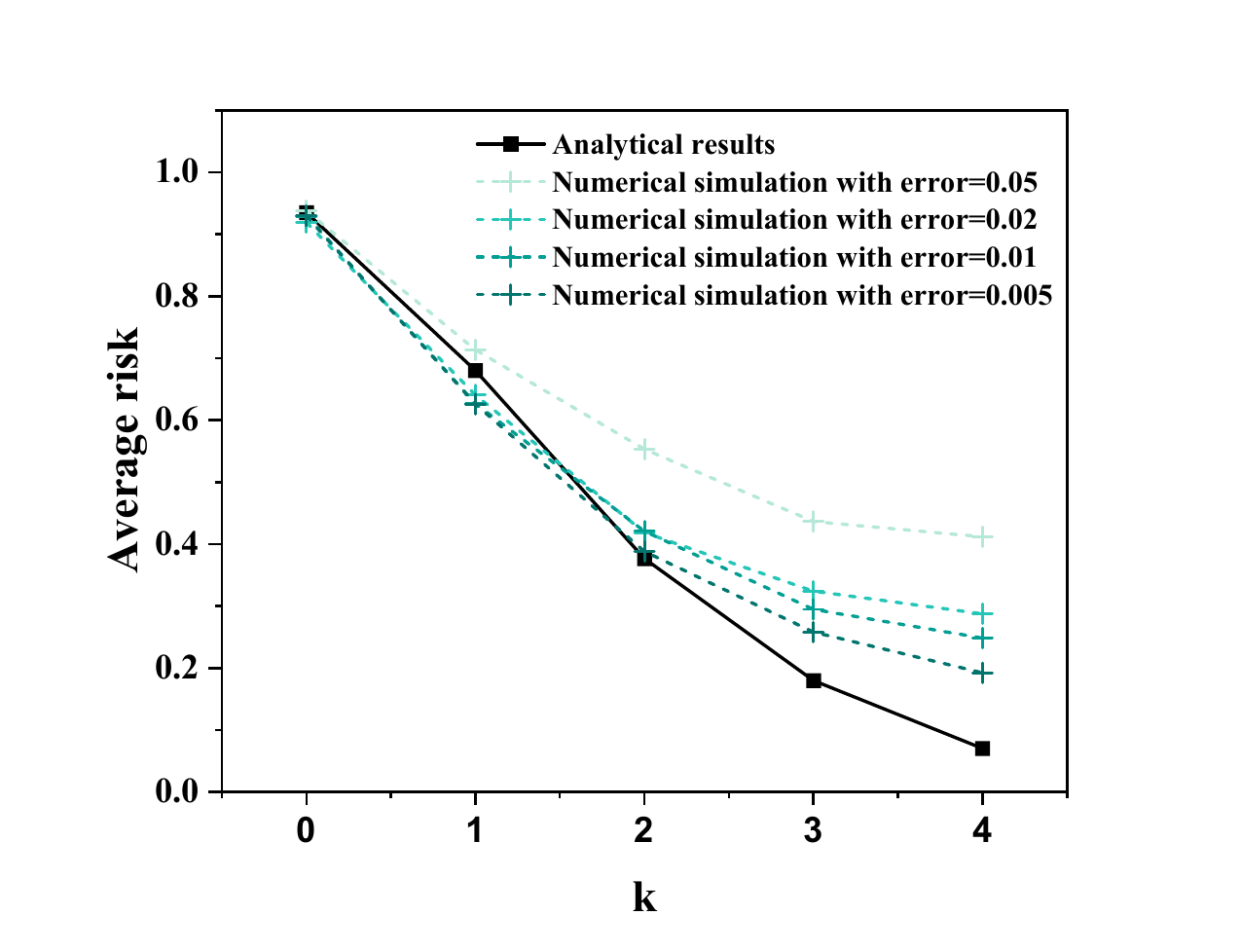}
	\caption{Average  risk  of the trained MPS-based machine learning models with respect to the training set size $t_k = 2^n-2^{n-k}$, where the system qubit size is $n=4$. The solid line represents the analytical lower bounds of the average risk. And the dotted  lines denote the average risk  of the numerical simulations with training error varying from 0.05 to 0.005.}
\label{convergence}
\end{figure}

\section{Review of no-free-lunch theorem in quantum machine learning models}

The exact formulations of no-free-lunch (NFL) theorems in quantum machine learning models have been introduced in Refs. \cite{Poland2020No,Sharma2022Reformulation}. Here for completeness, we straightforwardly consider the entanglement-assisted violation of quantum no-free-lunch theorem \cite{Sharma2022Reformulation}. One can first define the quantum Hilbert space for input samples by $\mathcal{H}_{xR}=\mathcal{H}_{x}\otimes\mathcal{H}_{R}$, and define the quantum Hilbert space for output by $\mathcal{H}_{yR}=\mathcal{H}_{y}\otimes\mathcal{H}_{R}$, where $\mathcal{H}_{R}$ denotes the ancillary Hilbert space. Without loss of generality,  we suppose that the dimensions of both the input and output space are the same, i.e., $\dim(\mathcal{H}_{xR})=\dim(\mathcal{H}_{yR})=N$. One can then choose a subset $\mathcal{S}_Q$ of the input samples as the training data for quantum machine learning, 
\begin{equation}
\mathcal{S}_Q=\left\{\left(|\psi_j\rangle,|\phi_j\rangle\right):|\psi_j\rangle\in \mathcal{H}_{xR},|\phi_j\rangle\in \mathcal{H}_{yR}\right\}_{j=1}^t,
\end{equation}
with the set size $|\mathcal{S}_Q|=t$, and $|\phi_j\rangle=(U\otimes I_R)|\psi_j\rangle$. 

One assumes that all training data states share the same Schmidt rank $r\in \{1,2,\cdots, d\}$.

For perfect quantum learning, the trained unitary $V_{\mathcal{S}_Q}$ on those training data states should result in the following formula
\begin{equation}
|\tilde{\phi_j}\rangle = (V_{\mathcal{S}_Q}\otimes I_R)|\psi_j\rangle = e^{i\theta_j}(U\otimes I_R)|\psi_j\rangle, \quad  \forall |\psi_j\rangle \in \mathcal{S}_Q.
\end{equation}

To quantify the accuracy of the trained unitary $V_{\mathcal{S}_Q}$, one can define the risk function
\begin{equation}
\begin{aligned}
R_{U}(V_{\mathcal{S}_Q}):=&\int dx\left\| U|x\rangle\langle x|U^{\dag}-V_{\mathcal{S}_Q}|x\rangle\langle x|V_{\mathcal{S}_Q}^\dag \right\|_1^2\\
=& \int dx \left(1-\left(\langle x|U^\dag V_{\mathcal{S}_Q}|x\rangle\right)^2\right)\\
=& 1-\int_{\rm Haar} dx \left|\langle x|U^\dag V_{\mathcal{S}_Q}|x\rangle\right|^2,
\end{aligned}
\end{equation}
where  $\|A\|_1 = \frac{1}{2}\tr[\sqrt{A^\dag A}]$ denotes the trace norm of $A$. Suppose that the input states $|x\rangle$ form the approximate unitary-2 design, i.e., approaching to the 2-moment integral of unitary group under the Haar measure. Then the risk function has the following formula
\begin{equation}
R_{U}(V_{\mathcal{S}_Q})=1-\frac{N+|\tr(U^\dag V_{\mathcal{S}_Q})|^2}{N(N+1)}.
\end{equation}

Consider the following three cases:
\begin{itemize}
\item[(a)] States in $\mathcal{S}_Q$ are orthonormal. $W=U^\dag V_{\mathcal{S}_Q} =e^{i\theta_1 I_r}\oplus e^{i\theta_2 I_r} \cdots \oplus  e^{i\theta_t I_r}\oplus Y$. $Y$ is the $SU(N-rt)$ group element. 
\item[(b)] States in $\mathcal{S}_Q$ are non-orthonormal, but linear independent. $W=U^\dag V_{\mathcal{S}_Q} =e^{i\theta I_{r\times t}} \oplus Y$. $Y$ is the $SU(N-rt)$ group element. 

\item[(c)] States in $\mathcal{S}_Q$ are  linear dependent. $W=U^\dag V_{\mathcal{S}_Q} =e^{i\theta I_{r\times t'}} \oplus Y$, $t'$ denotes the linear independent bases of the training set.  $Y$ is the $SU(N-rt')$ group element. 
\end{itemize} 

The goal is to calculate the average risk over all training sets ${\mathcal{S}_Q}$ and all (possible output) unitaries $U$ \cite{Sharma2022Reformulation},
\begin{equation}
\mathbb{E}_{U}\left[\mathbb{E}_{\mathcal{S}_Q}\left[R_{U}(V_{\mathcal{S}_Q})\right]\right].
\end{equation}

Consider the case (a): States in $\mathcal{S}_Q$ are orthonormal. One first calculates the average risk over all possible uniatries $U$,
\begin{equation}
\begin{aligned}
\mathbb{E}_{U}\left[R_{U}(V_{\mathcal{S}_Q})\right]=& 1-\frac{N+\int dU |\tr(U^\dag V_{\mathcal{S}_Q})|^2}{N(N+1)}\\
=& 1-\frac{N+ \int dY\left|r\sum_{j=1}^t e^{i\theta_j}+\tr[Y]\right|^2}{N(N+1)}\\
=& 1-\frac{N+  \int dY\left[r^2\left|\sum_{j=1}^t e^{i\theta_j}\right|^2+\left|\tr[Y]\right|^2+r\sum_{j=1}^t e^{i\theta_j}\tr[Y^\dag]+r\sum_{j=1}^t e^{-i\theta_j}\tr[Y]\right]}{N(N+1)}\\
\geq & 1-\frac{N+r^2t^2+\int dY \tr[Y]\tr[Y^*] }{N(N+1)} = 1-\frac{N+r^2t^2+1}{N(N+1)}.
\end{aligned}
\end{equation}
Since this is independent of the choice of training set, one can then obtain that
\begin{equation}
\mathbb{E}_{U}\left[\mathbb{E}_{\mathcal{S}_Q}\left[R_{U}(V_{\mathcal{S}_Q})\right]\right]\geq 1-\frac{N+r^2t^2+1}{N(N+1)}.
\end{equation}
Since $t$ can be arbitrary positive integer that is smaller than $d$, thus for the Schmidt number $r\geq2$, i.e., with non-zero bipartite entanglement between $\mathcal{H}_x$ and $\mathcal{H}_R$, such an average risk over all choices of training set and all possible ground truth unitaries $U$ can go beyond the classical no-free-lunch statements.  

Similar calculations apply to the cases (b) and (c). 

To summary, one obtains the quantum no-free-lunch theory for quantum machine learning models.
\begin{itemize}
\item Case (a): States in $\mathcal{S}_Q$ are orthonormal. 
\begin{equation}
\mathbb{E}_{U}\left[\mathbb{E}_{\mathcal{S}_Q}\left[R_{U}(V_{\mathcal{S}_Q})\right]\right]\geq 1-\frac{N+r^2t^2+1}{N(N+1)}.
\end{equation}
\item Case (b):   States in $\mathcal{S}_Q$ are non-orthonormal, but linear independent. $W=U^\dag V_{\mathcal{S}_Q} =e^{i\theta I_{r\times t}} \oplus Y$. $Y$ is the $SU(N-rt)$ group element. 
\begin{equation}
\mathbb{E}_{U}\left[\mathbb{E}_{\mathcal{S}_Q}\left[R_{U}(V_{\mathcal{S}_Q})\right]\right]= 1-\frac{N+r^2t^2+1}{N(N+1)}.
\end{equation}
\item Case (c):   States in $\mathcal{S}_Q$ are  linear dependent. $W=U^\dag V_{\mathcal{S}_Q} =e^{i\theta I_{r\times t'}} \oplus Y$, $t'$ denotes the linear independent bases of the training set.  $Y$ is the $SU(N-rt')$ group element. 
\begin{equation}
\mathbb{E}_{U}\left[\mathbb{E}_{\mathcal{S}_Q}\left[R_{U}(V_{\mathcal{S}_Q})\right]\right]= 1-\frac{N+r^2t'^2+1}{N(N+1)}.
\end{equation}
\end{itemize}

\end{document}